\documentclass{myElsart}

\usepackage[colorlinks=true,citecolor=blue,linkcolor=blue]{hyperref}
\usepackage[colorlinks=true,citecolor=blue,linkcolor=blue]{hyperref}
\usepackage{amsmath}
\usepackage{amssymb}
\usepackage{bbm}
\usepackage{cancel}
\usepackage{theorem}
\usepackage[]{natbib}
\usepackage{tikz}
\usetikzlibrary{arrows}
\usepackage{enumitem}

\newtheorem{theorem}{Theorem}[section]
\newtheorem{lemma}[theorem]{Lemma}
\newtheorem{proposition}[theorem]{Proposition}

\newtheorem{assumption}[theorem]{Assumption}

\newtheorem{definition}[theorem]{Definition}

\newenvironment{proof}[1][Proof]{\vspace{1em}\begin{trivlist}
\item[\hskip \labelsep {\bfseries #1}]}{\hfill $\Box$\end{trivlist}\vspace{1em}}

\newcommand\independent{\protect\mathpalette{\protect\independenT}{\perp}}
\def\independenT#1#2{\mathrel{\rlap{$#1#2$}\mkern5mu{#1#2}}}

\newcommand{\pr}{{\mathbb P}} 
\newcommand{\expt}{{\mathbb E}} 
\tikzset{
  treenode/.style = {inner sep=0pt, text centered,
    font=\sffamily},
  arn_n/.style = {treenode, circle, white, font=\sffamily\bfseries, draw=black,
    fill=black, text width=1em},
}

\begin{document}

\begin{frontmatter}

\title{Perturbation graphs, invariant causal prediction and causal relations in psychology}
\author{Lourens Waldorp, }\ead{waldorp@uva.nl}
\author{Jolanda Kossakowski, }
\author{Han L. J. van der Maas}
\address{University of Amsterdam, Nieuwe Achtergracht 129-B, 1018 NP, the Netherlands}

\begin{abstract}
Networks (graphs) in psychology are often restricted to settings without interventions. Here we consider a framework borrowed from biology that involves multiple interventions from different contexts (observations and experiments) in a single analysis. The method is called perturbation graphs. In gene regulatory networks, the induced change in one gene is measured on all other genes in the analysis, thereby assessing possible causal relations. This is repeated for each gene in the analysis. A perturbation graph leads to the correct set of causes (not necessarily direct causes). Subsequent pruning of paths in the graph (called transitive reduction) should reveal direct causes. We show that transitive reduction will not in general lead to the correct underlying graph. We also show that invariant causal prediction is a generalisation of the perturbation graph method, where including additional variables does reveal direct causes, and thereby replacing transitive reduction. We conclude that perturbation graphs provide a promising new tool for experimental designs in psychology, and combined with invariant causal prediction make it possible to reveal direct causes instead of causal paths. As an illustration we apply the perturbation graphs and invariant causal prediction to a data set about attitudes on meat consumption and to a time series of a patient diagnosed with major depression disorder. 
\end{abstract}

\end{frontmatter}
\endNoHyper
\newpage
\section{Introduction}\noindent
Although networks have become increasingly popular in psychology \citep{Borsboom:2011,Epskamp:2016,Borsboom:2017,Dalege:2017,Marsman:2018,Kan:2019,Waldorp:2019}, the connections between psychological variables are obtained from observational (i.e., no intervention) data \citep[but see][for an interesting alternative]{Blanken:2020}. 
However, in psychology it is commonplace to  perform controlled experiments in which  different conditions, including control or placebo conditions, are compared. 
Methods that make use of both experimental and observational data in the construction of psychological networks would therefore be a useful addition to the current instrumentarium. Here we investigate methods developed in biology for gene regulatory networks, where networks are constructed by combining different contexts, i.e., observation and intervention contexts. 

One popular method to obtain directed networks in biology involves the aggregation of observation (wild type) data and data from different experiments \citep{Markowetz:2005b,Markowetz:2007,Klamt:2010,Pinna:2010,Pinna:2013}. In each experiment one of the nodes is perturbed (intervened on) and then the effects on other nodes is measured. These effects on other nodes are assessed by determining non-zero correlations between nodes using both observation and intervention data combined (or mean differences between observation and intervention conditions using $t$-scores). The resulting graph is called a perturbation graph, where each edge represents a directed path that implies that changing the starting variable of the path leads to a change in the distribution at the end variable of the path. In the next step, the correlations are used to determine which of the of these paths could be direct connections. This step is called transitive reduction, named after a method in graph theory \citep{Aho:1972}. The two-step procedure results in a graph that is expected to represent direct (i.e., direct causal) effects between nodes. 

An example with three variables may help elucidate the idea of perturbation graphs. Suppose we have variables $s$, $u$ and $t$. Although we do not know this, we assume it is true that $s$ causes $u$ and $u$ causes $t$. Then if we were to intervene on $s$ we would learn that both $u$ and $t$ would change in distribution. Likewise, intervening on $u$ would lead to $t$ having a change in distribution but not $s$, and intervening on $t$ would not lead to any change. So, a perturbation graph tells us which variables are effected by interventions. Then transitive reduction, the subsequent step, is meant to determine that there is no direct path from $s$ to $t$, but only through $u$. (See sections \ref{sec:cond-cor} and \ref{sec:transitive-reduction} for details.)

In psychological research we can also implement such a procedure. 
For example, we could measure attitudes on eating meat (e.g., ``The production of meat is harmful for the environment'') and then present hypothetical scenarios (e.g., ``The meat and dairy industry has a huge CO-2 emission and is therefore harmful for the environment.'') and measure the same set of attitudes again. And so on, for each of the different attitudes \citep{Hoekstra:2018}. We get a collection of data sets where each of the variables has been intervened on by hypothetical scenarios, including observation without an intervention. By determining the conditional correlations and pruning some of the apparently superfluous connections (transitive reduction), we then obtain a causal graph. 

This two-step procedure, creating a perturbation graph and pruning away superfluous connections aligns well with psychology research in at least two ways. First, as in traditional research design in psychology, a control (observational) group is compared to (combined with) an experimental group in a single analysis. In the perturbation graph exactly such results are visualised, but then of multiple such experiments (comparisons). Second, the results of a perturbation graph allow for confirmatory research; hypotheses about the effects can be drawn up beforehand, upon which they can be tested. The strength of using the perturbation graph method is that by using multiple variables and considering multiple interventions at once, more information about possible pathways between the variables (mechanistic like) is available. 

Unfortunately, it turns out that both steps, creating the perturbation graph and pruning connections (transitive reduction), will not in general lead to the correct set of direct causes (we prove this in Sections \ref{sec:cond-cor} and \ref{sec:transitive-reduction}). In the first step, to arrive at a perturbation graph, observational and experimental data are pooled, as if these settings have similar properties. For nodes $s$ and $t$ an edge $s\to t$ in a perturbation graph represents the fact that some intervention on $s$ caused a change in $t$. We show, however, with a counter example, that using the pooled information, as suggested for perturbation graphs, does not always lead to the correct decision that there is an effect (we prove this in Section \ref{sec:cond-cor}). Furthermore, the second step, transitive reduction, may also lead to false conclusions. In the perturbation graph it is unclear if $s$ is a direct cause of $t$ or if there is a path $s\to\cdots\to t$. The transitive reduction step tries to resolve this by removing the direct connection $s\to t$ if its coefficient is small compared to all the coefficients along another path $s\to\cdots\to t$. We show that this need not hold in all cases, and so there is no guarantee to obtain the correct direct causes (we prove this in Section \ref{sec:trans-red}). The methodology of the perturbation graph is still useful though. We show that the perturbation graph method is a special case of a more general framework called invariant causal prediction \citep{Peters:2015}. Our contribution is to bring together the methodology of the perturbation graph, which aligns well with psychology research, and invariant causal prediction, which is a method for causal discovery.

Invariant causal prediction revolves around the idea that if we have the correct set of direct causal nodes then the conditional distribution of the node given its set of direct causes will be the same (i.e., is invariant) irrespective of any intervention on the non-direct causes \citep{Peters:2015,Meinshausen:2016,Magliacane:2016,Mooij:2020a}. This idea involves the seminal concept of considering different sets of nodes for the conditional distributions to obtain a causal graph, introduced by, e.g., \citet{Blalock:1960}, \citet{Spirtes:1993}, \citet{Pearl:1992}, \citet{Lauritzen:2001}, \citet{Lauritzen:2002},  \citet{Hyttinen:2012} and \citet{Tillman:2014}. For each context (observation or intervention) all sets of possible direct causes are considered and it is determined for each of those sets whether they are invariant across different contexts. 
It was shown that under relatively mild conditions (see Section \ref{sec:cond-inv-pred}), the correct set of direct causes can be obtained by this method \citep{Peters:2015,Meinshausen:2016}. The assumptions required in \citet{Peters:2015} to obtain the correct graph are that (a) in each setting (observational and experimental) the causal relations are the same, (b) the relations between the variables are linear, and (c) the error terms of the variables are independent of each other. Although these assumptions may appear strong, both (b) and (c) may be relaxed. We discuss these assumptions more elaborately in Section \ref{sec:graphs}. 

The advantages of using perturbation graphs with invariant causal prediction are: (1) multiple variables are considered as causes and effects using different contexts (experiment and observation), (2) variables can be of mixed type (e.g., continuous and discrete), and (3) direct causes may be revealed instead of only causal paths. We will illustrate this with an example data set on meat consumption and a time series of a single patient diagnosed with major depression disorder in \hyperref[sec:application]{Section \ref*{sec:application}}.

\subsection{Contribution and relation to other work}\noindent
In this work we focus on the idea that we want to combine observational and experimental data in a single analysis to improve estimation of the network, the direct causes between several variables. Our contribution is to bring together the work of perturbation graphs, mostly known in biology, and the work of causal discovery in statistics and machine learning. Specifically, we show that a perturbation graph as defined in the biology literature is a marginal version of the causal discovery version of invariant causal prediction. This framework of invariant causal prediction is necessary, since the two-step procedure of the standard perturbation graph and transitive reduction cannot guarantee a correct solution (we prove this in Sections \ref{sec:cond-cor} and \ref{sec:transitive-reduction}).

Several variations of algorithms to perform the transitive reduction step have been proposed. For instance, \citet{Rice:2005} use the so-called conditional correlation, which is a pairwise correlation obtained by aggregating observational (wild-type) and experimental (manipulation) data (we discuss this in detail in Section \ref{sec:cond-cor}). This idea was extended and improved by \citet{Klamt:2010}, where additional constraints on retaining edges were imposed. \citet{Markowetz:2005b} proposed the so-called nested effects model, where effects of genes on other genes are mediated by the up- or downstream role of genes in the network. Other suggestions include determining the dynamic structure of gene networks \citep{Frohlich:2011,Anchang:2009}. More recently, \citet{Shojaie:2014} proposed an algorithm based on scoring rules from lasso regressions and order sorting algorithms; this algorithm is particularly suited for large-scale graphs where the only question is which nodes are connected but not the direction. Finally, a logic based algorithm was introduced \citep{Gross:2019} to cope with large networks and robust inference. However, none of these methods solves the problem with transitive reduction (see Section \ref{sec:transitive-reduction}). 

There are several methods to obtain causal effects from observational data alone. The original idea by \citet{Pearl:1991} and \citet{Spirtes:1991} uses only observational data that lead to Markov equivalent models, where not all causal relations can be identified \citep[see e.g., ][]{Geiger:1990,Pearl:2009}. Alternatively, by assuming that the variables are non-Gaussian, it is possible to identify more causal relations than assuming Gaussian or multinomial random variables; the method is referred to as LinGaM \citep{Shimizu:2006}. Another approach using only observational data is to use non-linear models \citep{Mooij:2009r} which also improves identifiability under certain assumptions \citep[see][for an excellent discussion]{Eberhardt:2017}. 

Mediation analysis can be seen a special case of causal discovery. In a mediation analysis for three variables it is assumed known that there are two causal paths to an outcome variable variable, and only a change in regression weight of the direct path is considered with respect to third mediating variable \citep[see, e.g.,][]{Baron:1986}. In contrast, in causal discovery with three variables such changes are considered for each pair of variables, since it is not assumed known which directions the paths might have. 

Here we focus on methods for causal discovery that combine both observation and intervention data \citep{Tillman:2014}. The combination of the two types of contexts is a combination of theory-driven and a more explorative approach to research. In the linear Gaussian setting the work by \citet{Hyttinen:2012} provides conditions where at the population level causal relations using both observation and (`surgical') intervention data obtain the correct underlying acyclic graph. \citet{Peters:2015} generalised the settings to different types of intervention, which is what we discuss here, and they call it invariant causal prediction \citep[see also ][]{Meinshausen:2016}. \cite{Mooij:2020a} generalised causal discovery even more by creating indicator variables for the setting (observation or any combination of other design factors) and making them part of the graph. In previous work \citep{Kossakowski:2021} we compared in simulations the perturbation graphs combined with transitive reduction and invariant causal prediction (see also Section \ref{sec:application}). Here our contribution is to bring together perturbation graphs and causal invariant prediction and determine possible benefits for research in psychology.

We begin in \hyperref[sec:graphs]{Section \ref*{sec:graphs}} by introducing graphs and the assumptions of the methods. Then in \hyperref[sec:cond-cor]{Section \ref*{sec:cond-cor}} we introduce the conditional correlation and  the marginal version of invariant prediction. In \hyperref[sec:transitive-reduction]{Section \ref{sec:transitive-reduction}} we introduce transitive reduction and show that transitive reduction is in general not consistent for the true underlying graph. Next, we describe invariant causal prediction in \hyperref[sec:cond-inv-pred]{Section \ref{sec:cond-inv-pred}}. Then in \hyperref[sec:application]{Section \ref{sec:application}} we illustrate the combination of perturbation graphs and invariant causal prediction with a data set from social psychology and psychopathology.

\section{Graphical models and interventions}\label{sec:graphs}\noindent
We use the language of graphical models, which ties together a graph with nodes and edges and random variables \citep[e.g.,][]{Lauritzen96,Cowell:1999,Eberhardt:2017,Maathuis:2018}. A directed graph is denoted by $\mathcal{G}$ containing a set of vertices $V=\{1,2,\ldots,m\}$ (also called nodes) and a set of directed edges $E=\{i\to j: i,j\in V\}$ (also called arrows). The graphs we consider have at most a single edge between each pair of nodes. An example of a graph is shown in Figure \ref{fig:d-separation-example}(a). In the configuration $s\to t$, node $s$ is called a parent of $t$ and $t$ is called a child of $s$. The set of nodes $\textsc{pa}(t)$ contains all nodes that are parents of $t$, and the set of nodes in $\textsc{ch}(s)$ contains all children of $s$. A parent is sometimes referred to as a direct cause. A path is a sequence of (at least two) distinct nodes where each node is either a parent or a child of the next node in the sequence. For example, $s\to u\leftarrow t$ is a path. A directed path from $s$ to $t$ is a path such that all edges point in the same direction. For example, $s\to u\to t$ is a directed path. A cycle is a directed path where each node occurs once except the first node. For example, $s\to u\to t\to s$ is a cycle. Node $s$ is called an ancestor of $t$ if there is a directed path from node $s$ to node $t$, as in $s\to\cdots\to t$; the set of nodes $\textsc{an}(t)$ is the set of ancestors of $t$, where we additionally include the node $t$. The configuration $s\to u\leftarrow t$ is called a collider path.

The graphs that we consider are directed acyclic graphs (DAGs). A DAG is a graph with only directed edges and no cycles. DAGs have been frequently studied and have been used as the basic building block in many models \cite[e.g., ][]{Spirtes:1993,Buhlmann:2011,Rothenhausler,Magliacane:2016,Eberhardt:2017,Maathuis:2018}. 
%
%

We use the causal Markov condition to connect graphs and probability distributions using the relation between $d$-separation and conditional independence, as described in \citet{Pearl:2000}, \citet{Pearl:1992}, \citet{Spirtes:1993}, \citet{Peters:2017} and \citet{Maathuis:2018}. Using this allows us to identify the nodes in a graph with the random variables $X=(X_1,X_2,\ldots,X_m)$, which together constitute a graphical model. The edges in a graphical model are defined by an association, like partial correlations for multivariate normal variables.

We use the version of $d$-separation of \citet[][Proposition 3.25]{Lauritzen96}, where a directed graph is transformed to an undirected graph in three steps. First, we obtain the relevant nodes for subset $A$ in which all ancestors of $A$ are included (and so also $A$ itself). In the second step for each two parents of the same child a new edge is introduced. In the third step all edges are made undirected. Then, $d$-separation is defined as follows. Let subsets of nodes $A$, $B$ and $C$ be disjoint. Nodes in $A$ and $B$  are $d$-separated by nodes in $C$ if all paths from nodes in $A$ to nodes in $B$ go through nodes in $C$; if there are no paths from nodes in $A$ to nodes in $B$, then $A$ and $B$ are $d$-separated. 

As an example, consider the directed graph in Figure \ref{fig:d-separation-example}(a). The relevant (ancestral) variables of $A=\{s,t,u\}$ contains the nodes $\{s,t,u,v\}$, the same edges as in the original graph. Next, an additional edge $u-v$ is introduced to the parents $u$ and $v$ of child $s$ (see Figure \ref{fig:d-separation-example}(b)). Now, $u$ and $t$ can be seen to be $d$-separated by $s$ because any path from $u$ to $t$ goes through $s$. And similarly, nodes $v$ and $t$ are $d$-separated by $s$ since all paths from $v$ to $t$ go through $s$. And so $s$ blocks the paths from $u$ and $v$ to $t$. We denote that nodes in $A$ are $d$-separated from nodes in $B$ by $C$ as $A\independent B\mid C$.  More details about the  assumptions and related concepts can be found in Appendix \ref{sec:appendix-gm-assumptions}.
\begin{figure}[t]\centering
\begin{tabular}{c @{\hspace{5em}} c}

\begin{tikzpicture}[->,auto,node distance=2cm,
  thick,main node/.style={circle,draw,font=\footnotesize\sffamily}]
  \tikzstyle{sample} = [circle,draw,font=\footnotesize\sffamily,minimum size=2.5em]
  \tikzstyle{sampleEdge} = [font=\sffamily\small]
	
  \node[sample] (t) {$t$};
    \node[sample] (s) [above left of=t] {$s$};
    \node[sample] (u) [below left of=s] {$u$};
  \node[sample] (v) [above left of=u] {$v$};
 \path[sampleEdge,<-]
    (t) edge node [left] {} (s); 
 \path[sampleEdge,->]
    (u) edge node [left] {} (s) ; 
 \path[sampleEdge,<-]
    (s) edge node [left] {} (v) ; 
\end{tikzpicture}

&

\begin{tikzpicture}[->,auto,node distance=2cm,
  thick,main node/.style={circle,draw,font=\footnotesize\sffamily}]
  \tikzstyle{sample} = [circle,draw,font=\footnotesize\sffamily,minimum size=2.5em]
  \tikzstyle{sampleEdge} = [font=\sffamily\small]
  \tikzstyle{noSample} = [circle,font=\footnotesize\sffamily,minimum size=2.5em]
	
  \node[sample] (t) {$t$};
    \node[sample] (s) [above left of=t] {$s$};
    \node[sample] (u) [below left of=s] {$u$};
  \node[sample] (v) [above left of=u] {$v$};
 \path[sampleEdge,-]
    (v) edge node [left] {} (u) ; 
 \path[sampleEdge,-]
    (s) edge node [left] {} (t) ; 
 \path[sampleEdge,-]
    (s) edge node [left] {} (u) ; 
 \path[sampleEdge,-]
    (v) edge node [left] {} (s) ; 
\end{tikzpicture}

\\

(a) & (b)

\end{tabular}
\caption{(a) A directed acyclic graph (DAG) with a collider path $v\to s\leftarrow u$. (b) The ancestral moralised graph $\mathcal{G}_{\{s,t,u\}}^m$, where the parents $\textsc{pa}(s)=\{u,v\}$ are connected and all arrows are removed and so the graph is undirected. The nodes $u$ and $t$ and $v$ and $t$ are $d$-separated by $\{s\}$.  } 
\label{fig:d-separation-example}
\end{figure}

It is often assumed that all relevant variables are in the analysis of the system under investigation. This is  referred to as the assumption of \textit{causal sufficiency} \citep{Eberhardt:2017}. The assumption is relevant as there could be a correlation (dependence) between variables because of links with a variable outside the variables considered (unobserved variables). We discuss this issue in Section \ref{sec:unobserved}.

We estimate the network in a nodewise fashion, as in \citet{Meinshausen:2006} and \citet{Pircalabelu:2015}. In turn we let each node be the target (dependent) variable in a linear regression with Gaussian noise and determine the neighbourhood of each node from the non-zero coefficients in the regression. The support of the target variable $t$ is a subset of the remaining variables for which the edges are non-zero.
Here, we will restrict attention to linear relations between variables and errors that are Gaussian and uncorrelated. Linearity simplifies many of the ideas and the Gaussian assumption is quite common in the social sciences \citep[e.g.,][]{Koster:1996}.
\begin{assumption}{\em (Linear Gaussian model)}\label{ass:linear-gaussian}
Let $X_{t}$ for any $t\in V$ be a random variable with values in the set of real numbers $\mathbb{R}$ and let $\varepsilon_{t}$ be normally (Gaussian) distributed with mean zero and variance $\sigma^{2}_{t}$ and uncorrelated with any other random variable $\varepsilon_s$ for $s\in V\backslash \{t\}$ (i.e., $\textrm{cov}(\varepsilon_{t},\varepsilon_{s})=\sigma^{2}_{t}$ if $s=t$ and 0 otherwise).  Then the \emph{linear Gaussian model} is
\begin{align}\label{eq:true-model}
X_{t} = \sum_{s\in V\backslash \{t\}} X_{s}\beta_{ts}+\varepsilon_{t}
\end{align}
where $\beta_{ts}$ denotes the coefficient for predictor $X_s$ in the regression for $X_t$.  
\end{assumption}
\citet{Verma:1991} show that using the linear Gaussian model (\ref{eq:true-model}) induces a probability distribution that is Markov with respect to the graph (in fact, they show a more general result which includes this, see also \citet[][Proposition 6.31]{Peters:2017}). Hence, we could suffice with Assumption \ref{ass:linear-gaussian}, which then implies the Markov condition in Assumption \ref{ass:markov}. In any case, for linear Gaussian models we obtain the Markov condition. If we do not assume linearity and independent errors, then we must assume the Markov condition holds, so that $d$-separations of the graph imply conditional independencies in the probability distribution (see Appendix \ref{sec:appendix-gm-assumptions}). 

We are interested in both observational and interventional (experimental) data. Therefore, we need to specify what it means in terms of graphical models to induce an intervention. Here we will use hard and soft interventions \citep[e.g.,][]{Eberhardt:2007}, although more types of interventions exist \citep[see, e.g.,][for an overview]{Mooij:2020a}. In \hyperref[sec:appendix-interventions]{Appendix \ref*{sec:appendix-interventions}} we give precise definitions and provide here a brief description of hard and soft interventions. A hard intervention \citep[e.g.,][]{Pearl:2009} can be considered as an intervention that completely takes over the control of a variable, and so no other variables can affect it. Graphically, a hard intervention is considered as the removal of any effect on the variable, i.e., arrows into the variable are removed. It is sometimes referred to as intervention by replacement \citep{Lauritzen:2002} because a hard intervention replaces the original value (obtained from effects of other variables) by a new one (random or fixed) dictated completely by the experimenter. For example, in a visual perception experiment the level of contrast that can be seen by a participant is a hard intervention. In contrast, a soft intervention leaves the structure of the graph intact and so the effects other nodes have on the intervened node remain. A resulting effect of a soft intervention could, for instance, be a change in mean or variance, or both. But, importantly, the effects of other variables remains present in a soft intervention. Soft interventions are mainly associated with quasi-experimental designs. For instance, a non-randomised design where participants can choose the treatment to a disease is a soft intervention. 

As opposed to normal conditioning, for intervention conditioning we use the notation $p(x\mid\mid C=c)$ for the density $p$ determined by the specification of the context variable $C=c$ \citep{Lauritzen:2002}. The variable $C$ can be considered as a random variable at the nominal level, indicating the type and index of intervention \citep{Eberhardt:2007}.  \citep[See][for different versions of defining a context.]{Mooij:2020a} As explained in Appendix \ref{sec:appendix-interventions}, the type of intervention or no intervention determines what the form of the probability distribution is. For example, in a hard intervention on $s$ (replacing $X_s$ by a random variable $W$), all terms that are parents of node $s$ are removed from the density \citep{Pearl:2000}.  We write $C=s$ to denote that a (hard or soft) intervention is applied to node $s$ and we write $C=\emptyset$ to indicate no intervention and so the observation distribution obtains,  i.e., $p(x\mid\mid C=\emptyset)=p(x\mid C=\emptyset)$.

\section{Conditional correlation, perturbation graph and invariance}\label{sec:cond-cor}\noindent
From observational data the correlations reveal where possible paths between nodes are. 
Reichenbach's (common cause) principle \citep{Reichenbach:1991} about non-zero correlations is in this respect well known: a path can take one of three forms in $\mathcal{G}$ which gives rise to the non-zero correlation \citep{Spirtes:1993,Peters:2017}: $s\to \cdots\to t$ or $s\leftarrow \cdots \leftarrow t$ or $s\leftarrow \cdots\to t$ (see \hyperref[lem:direct-path]{Lemma \ref*{lem:direct-path}}). 
From purely observational data we cannot distinguish between these configurations in linear models \citep{Eberhardt:2017}. By including information from interventions we may be able to distinguish between the three possible configurations. 
Given a non-zero correlation between nodes $s$ and $t$, an intervention on $s$ will lead to a corresponding change in the distribution of $t$, only if there is some path $s\to\cdots\to t$ (see \hyperref[lem:wiggle-s]{Lemma \ref*{lem:wiggle-s}}).
This idea is used in what are called perturbation graphs \citep{Klamt:2010} or response graphs \citep{Gross:2019}.

Perturbation graphs are quite popular in biology, and especially in gene regulatory networks (where the function of the interactions of genes is investigated). The reason is that it is possible to consider the effects of, for example, the absence of a specific gene (in knock-out mice) and determine the phenotypes \citep[see, e.g.,][]{Pinna:2013}. Different sets of genes can be made absent in order to investigate the effect of each gene within the gene regulatory network. A set of such experiments is valuable because not only can the functions of single genes be mapped to phenotypes, but also the interaction with other genes can be considered in such networks.  

These ideas, using sets of different experiments on different variables described above, are shown in Figure \ref{fig:scatterplots-interventions} for the two small graphs $s\to u\to t$ and $s\leftarrow u\to t$, with associated variables $X_{s}$, $X_{u}$ and $X_{t}$ (see \hyperref[sec:appendix-simulations]{Appendix \ref*{sec:appendix-simulations}} for the \textsf{R} code and \hyperref[sec:appendix-calculations-correlations]{Appendix \ref*{sec:appendix-calculations-correlations}} for calculations of correlations). In each of the plots both the observation context is shown (in red) and the hard intervention context (in blue) in terms of a scatterplot and the univariate density plots in the margin. We consider both the densities (conditional distributions) and the scatterplots (regressions).
\begin{figure}[t]\centering
\hspace{2em}$s\to u\to t$  \hspace{11em} $s\leftarrow u\to t$   \\[0.5em]
\pgfimage[width=.45\textwidth]{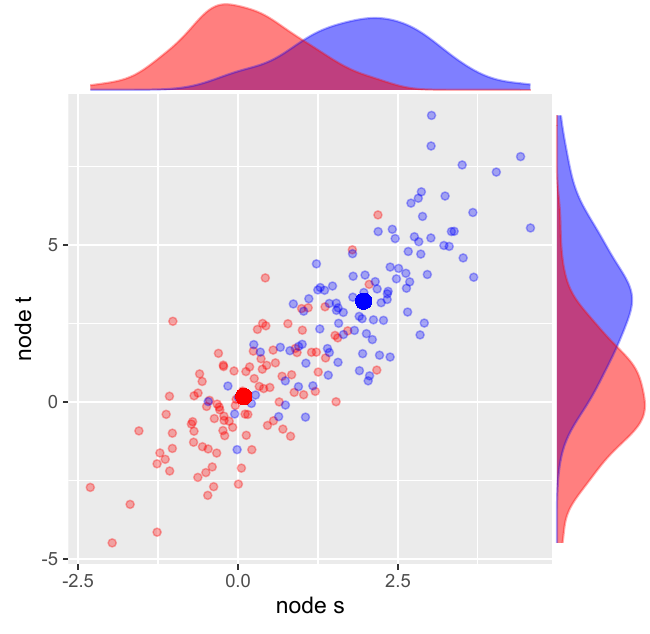}
~
\pgfimage[width=.45\textwidth]{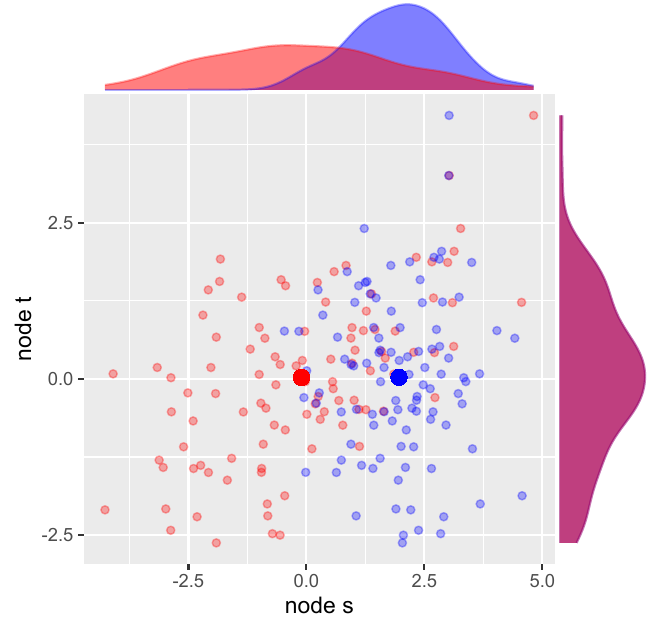}\\[0em]
(a)  \hspace{14.3em} (b)  
\caption{ 
Scatterplots for $X_{s}$ and $X_{t}$ in different contexts, where ${\color{red} C=\emptyset}$ (red) denotes observational data and and ${\color{blue} C=s}$ (blue) denotes interventional data, and the large filled circles are the corresponding empirical means. In (a) the underlying graph is $s\to u\to t$ and in (b) the underlying graph is $s\leftarrow u\to t$. The intervention is the hard intervention $s$ where $X_s$ by $W=2 + N(0,1)$.}
\label{fig:scatterplots-interventions}
\end{figure}

In Figure \ref{fig:scatterplots-interventions}(a) we see for the graph $s\to u\to t$ that without intervention (red) the means of nodes $s$ and $t$ are approximately 0. We apply a hard intervention on $s$ by replacing $X_{s}$ by $W=2+N(0,1)$, where $N(0,1)$ is a normal random variable with mean 0 and variance 1. Then the mean of nodes $s$ and $t$ are shifted upward because of $W$ (see Appendix \ref{sec:appendix-calculations-correlations}). The fact that a change is obtained in node $t$ by intervening on $s$ implies a directed path $s\to \cdots \to t$ (although we do not know if $s$ is a direct cause of $t$). In contrast, for the graph $s\leftarrow u\to t$ in Figure \ref{fig:scatterplots-interventions}(b), a hard intervention on $s$ results in no change in the distribution of node $t$. This leads to the conclusion that there \textit{cannot} be a directed path $s\to \cdots \to t$ (see Lemma \ref{lem:wiggle-s}).


This idea can be used to create a graph where a directed edge could be drawn to indicate a change in distribution from the observation context to the intervention context. An edge $s\to t$ in such a graph could indicate that $s$ is a direct cause (parent) of $t$ or an ancestor of $t$. Such a graph is called a perturbation graph \citep{Klamt:2010,Pinna:2013}.
\begin{definition}{\rm (Perturbation graph)}\label{def:perturbation-graph}
A {\em perturbation graph} $\mathcal{G}^{p}$ is a graph where any edge $s\to t$ represents that intervening on $s$ leads to a change in the distribution of $t$. 
\end{definition}
In the example graph $s\to u\to t$ above, the edge $s\to t$ would be included in the perturbation graph $\mathcal{G}^{p}$ because intervening (perturbing) on $s$ leads to a change in the distribution of $t$. After intervening on each node in turn, the complete perturbation graph $\mathcal{G}^{p}$ would then be $t \leftarrow s\to u\to t$. 

For better understanding we contrast the interpretation of an edge in a perturbation graph with an edge in a causal DAG. An edge in a perturbation graph represents a causal path, like in the example above: the edge $s\to t$ represents the causal path $s\to u\to t$. In a causal DAG the edge $s\to t$ represents the fact that $s$ is a direct cause of $t$ (assuming causal sufficiency, i.e. all relevant variables are measured). This is the main reason why a second step is required for the perturbation graph to prune away the edge $s\to t$, when there is a path $s\to u\to t$ (transitive reduction).  

To determine a change in distribution the mean or some type of correlation is often used \citep{Rice:2005,Klamt:2010}. Interestingly, the example above shows that invariance across different contexts in distribution suggests that the directed edges are similar in different contexts (observation and intervention). This is the idea of invariance across contexts used by \citet{Peters:2015} and \citet{Meinshausen:2016}, although in a more general form that we will describe in Section \ref{sec:cond-inv-pred}.
\begin{definition}{\rm (Marginal invariant prediction)}\label{def:marginal-invariant-prediction} Let $X_{s}$ with values in $\mathbb{R}$ be the single predictor obtained in contexts $C=\emptyset$ or $C=s$ with an intervention on $s$.Then, linear prediction with $X_{s}$ is called {\em marginal invariant predictive} for $X_{t}$ if there is a non-zero $\beta_{ts}\in\mathbb{R}$ such that
\begin{align}\label{eq:marginal-invariance}
X_{t}=X_{s}\beta_{ts}+ \varepsilon_{t},\quad \text{for contexts}\quad C=\emptyset \text{ and } C=s
\end{align}
and $\varepsilon_{t}$ is Gaussian with mean 0 and variance $\sigma^{2}_t$ and $\varepsilon_{t}\independent X_{s}$. 
\end{definition}
Here $X_i\independent X_j\mid X_k$ means that random variables $X_i$ and $X_j$ are independent. Because we have the same coefficient in different contexts, this invariance across different contexts suggests that we can compute the correlation using pooled data from the different contexts. This is the idea used in perturbation graphs \citep[e.g., ][]{Klamt:2010,Pinna:2010}, where non-zero correlations obtained from pooled data are used to determine possible directed paths of the underlying graph. 

We show in Proposition \ref{prop:condcor-invpred} that (\ref{eq:marginal-invariance}) implies that we can pool the data using the conditional correlation. We can define the conditional correlation $\rho_{ts \mid\mid s}$ as the correlation between nodes $t$ and $s$ in the context of an intervention on node $s$, i.e., $C=s$. This means that we use the conditional distribution of $X_t\mid\mid C=s$ to determine the moments, like the mean and covariance.  Then, when pooling data from both $C=s$ and $C=\emptyset$, we obtain the conditional correlation, denoted by $\rho_{ts\mid \mid \{\emptyset,s\}}$. This conditional correlation is based on a mixture distribution of different contexts \citep[see Appendix \ref{sec:appendix-conditional-correlation}; interpreted from][]{Rice:2005}. It follows that if we assume invariance then we can use the pooled information from both observational and interventional contexts to obtain a correlation, but not vice versa. (A proof is in Appendix \ref{sec:appendix-proofs}.)

\begin{proposition}\label{prop:condcor-invpred}
Assume the Markov assumption \ref{ass:markov} and the faithfulness assumption \ref{ass:faith} for a DAG $\mathcal{G}$ and distribution $\mathbb{P}$. Suppose the variables in graph $\mathcal{G}$ are measured in contexts $C=\emptyset$ and $C=s$ of equal probability with an intervention on $s$, and the variances of the nodes in $\mathcal{G}$ are the same across contexts. Then the following are equivalent
\begin{itemize}
\item[(i)]  $s$ is marginal invariant predictive for $t$, and
 \item[(ii)] the correlations without and with intervention are equal, i.e., $\rho_{ts\mid\mid \emptyset}=\rho_{ts\mid\mid s}$.
\end{itemize}
As a consequence, if $\rho_{ts\mid\mid \emptyset}=\rho_{ts\mid\mid s}$, then this also equals $\rho_{ts\mid\mid \{\emptyset,s\}}$
\end{proposition}
\hyperref[prop:condcor-invpred]{Proposition \ref*{prop:condcor-invpred}} makes precise that using the conditional correlation requires the assumptiopn of invariance (\ref{eq:marginal-invariance}), and unites the frameworks of invariant causal prediction by \citet{Peters:2015} and the framework of perturbation graphs \citep{Rice:2005,Klamt:2010}. 

Denote by $I_{\{\emptyset,s\}}$ the index set of observations associated with an equal number of data points for the observational and interventional contexts, where node $s$ was intervened on. The sample estimate of the conditional correlation is then \citep{Rice:2005}
\begin{align}\label{eq:cond-cor-sample}
r_{ts\mid\mid \{\emptyset,s\}}=\frac{\sum_{i\in I_{\{\emptyset,s\}}} (x_{s,i}-\bar{x}_{s})(x_{t,i}-\bar{x}_{t})}{\sqrt{\sum_{i\in I_{\{\emptyset,s\}}}(x_{s,i}-\bar{x}_{s})^{2}}\sqrt{\sum_{i\in I_{\{\emptyset,s\}}}(x_{t,i}-\bar{x}_{t})^{2}}}
\end{align}
In empirical work, the conditional correlation is computed for several different interventions. Each time there is a large conditional correlation between $s$ and $t$, the edge $s\to t$ is drawn, resulting in a set of directed paths called a perturbation graph \citep[see, e.g.,][]{Pinna:2010,Pinna:2013}. 

It is tempting to think that a non-zero conditional correlation suggests there is invariance and hence, that there is evidence for a directed path (causal relation). However, because the conditional correlation is based on a mixture distribution (see Appendix \ref{sec:appendix-conditional-correlation}), it is not in general true that a non-zero conditional correlation implies invariance; there is a counter example (thanks to an anonymous reviewer). 
If we consider the graph $s\leftarrow u\to t$ in Figure \ref{fig:scatterplots-interventions}(b), then the conditional correlation between $X_s$ and $X_t$ is approximately 0.329, thus suggesting a directed path $s\to\cdots\to t$. However, the correlation in the observational context is approximately 0.562, whereas in the interventional context the correlation is -0.003, and so $\rho_{ts\mid\mid \emptyset}\ne \rho_{ts\mid\mid s}$ for the graph $s\leftarrow u\to t$. This implies that obtaining a non-zero conditional correlation cannot provide evidence for a directed path $s\to\cdots \to t$. 

By considering the correlations from each context separately, we can obtain evidence that there is a directed path, and obtain a perturbation graph. The next step is to prune away some of the edges, since not all edges in the perturbation graph are direct causes. 

\section{Transitive reduction and invariant prediction}\label{sec:transitive-reduction}\noindent
We discuss two ways of determining whether a directed path $s\to \cdots\to t$ also implies that $s\to t$, i.e., whether $s$ is a parent (direct cause) of $t$. The first is called transitive reduction and is based on a heuristic criterion derived from a mathematical concept with the same name. The second is called causal invariant prediction and is based on the conditional distribution taking into account other variables. 
We start with transitive reduction in \hyperref[sec:trans-red]{Section \ref*{sec:trans-red}} and show that in general there is no guarantee that an accurate representation of the underlying graph is obtained. Then we will discuss the  invariant causal prediction method in \hyperref[sec:cond-inv-pred]{Section \ref*{sec:cond-inv-pred}} which does obtain an accurate representation of the underlying graph.

\subsection{Transitive reduction}\label{sec:trans-red}\noindent
The main idea of transitive reduction is that a direct cause (parent) should have a large contribution to the correlation; if there is no evidence of a large direct contribution, then there should not be a direct cause \citep{Rice:2005}. This version of pruning connections to determine whether there is a direct effect is directly related to the idea in graph theory of transitive reduction. In graph theory, the graph $\mathcal{G}_{2}$ is a transitive reduction of graph $\mathcal{G}_{1}$, where both graphs have the same set of nodes, if (i) $\mathcal{G}_{2}$ has a directed path from any node $s$ to a node $t$ only if $\mathcal{G}_{1}$ has such a path, and (ii) there is no other graph with fewer edges than $\mathcal{G}_{2}$ satisfying (i) \citep{Aho:1972}. This mathematical construction of a graph with similar connectivity properties does not intuitively lead to an accurate reconstruction of direct causes of nodes in a graph. We will show that, indeed, there is no guarantee with the method of transitive reduction that an accurate representation of the underlying graph is obtained. 


A result by \citet[][p. 567]{Wright:1921} shows that the correlation between $s$ and $t$ is obtained from different paths between $s$ and $t$ where each path $\alpha$ for $s\to \cdots\to t$ is the product of the correlations \citep[path coefficients,][]{Wright:1934}
\begin{align}\label{eq:path-coefficient-product}
d_{\alpha}=\rho_{v_{1}s}\rho_{v_{2}v_{1}}\rho_{v_{3}v_{2}}\cdots \rho_{tv_{k-1}}
\end{align}
(see \hyperref[sec:appendix-path-analysis]{Appendix \ref*{sec:appendix-path-analysis}} on path analysis for details).
%
Take, for example, the graph in Figure \ref{fig:wright-graph}(a) with coefficients $\rho_{ij}$. 
{\color{blue} In Figure \ref{fig:wright-graph}(a) path $\alpha_{1}$ is $s\to v\to t$ and path $\alpha_{2}$ is $s\to w\to t$}. 
The contribution to the correlation of the top path is $d_{\alpha_{1}}=\rho_{vs}\rho_{tv}$ and the contribution of the bottom path is $d_{\alpha_{2}}=\rho_{ws}\rho_{tw}$, which gives the correlation between $s$ and $t$, i.e., $\rho_{st}=d_{\alpha_{1}}+d_{\alpha_{2}}$. It follows that the correlation between $s$ and $t$ for the data given in Figure \ref{fig:wright-graph}(a) is 
\begin{align*}
\rho_{st}=\rho_{sv}\rho_{vt}+\rho_{sw}\rho_{wt}=(-0.2)(-0.2) + (-0.2)(-0.2)=0.08
\end{align*}
Here we see that the smallest absolute coefficient on the path between $s$ and $t$ is larger than the correlation between $s$ and $t$, i.e., $|\rho_{sv}|=0.2>0.08=|\rho_{st}|$. 
This gives rise to the heuristic criterion that if all of the conditional correlations $|\rho_{v_{j-1}v_{j}\mid\mid \{\emptyset,s\}}|$ (in absolute value) along the path between $s$ and $t$ are $>|\rho_{st\mid\mid \{\emptyset,s\}}|$, then it is likely that the correlation $\rho_{ts\mid\mid \{\emptyset,s\}}$ is induced by an alternative path and there is no direct connection $s\to t$, and it should be removed from the perturbation graph. \citet{Rice:2005} suggest that the criterion
\begin{align}\label{eq:tr-criterion}
\min\{|\rho_{sv_{1}\mid\mid \{\emptyset,s\}}|,\ldots,|\rho_{v_{k}t\mid\mid \{\emptyset,s\}}|\}>|\rho_{st\mid\mid \{\emptyset,s\}}|
\end{align}
is a good heuristic and it is clear that criterion (\ref{eq:tr-criterion}) is satisfied for the graph in Figure \ref{fig:wright-graph}(a). 

Unfortunately, criterion (\ref{eq:tr-criterion}) does not hold in general.
Consider the following counter example. For Figure \ref{fig:wright-graph}(b) there is a third path $s\to t$ with coefficient $-0.1$. The correlation between $s$ and $t$ is now $\rho_{ts}=(-0.20)(-0.20) + (-0.20)(-0.20)+ (-0.10)=-0.02$. This is closer to 0 than the minimum of any of the coefficients, and hence criterion (\ref{eq:tr-criterion}) is satisfied but leads to the wrong conclusion that parent $s$ should be removed from the set of possible causes of $t$. 
\begin{figure}[t]\centering
\begin{tabular}{c @{\hspace{5em}} c}

\begin{tikzpicture}[->,auto,node distance=2cm,
  thick,main node/.style={circle,draw,font=\footnotesize\sffamily}]
  \tikzstyle{sample} = [circle,draw,font=\footnotesize\sffamily,minimum size=2.5em]
  \tikzstyle{sampleEdge} = [font=\sffamily\small]
	
  \node[sample] (t) {$t$};
  \node[sample] (v) [above left of=t] {$v$};
  \node[sample] (s) [below left of=v] {$s$};
  \node[sample] (w) [below left of=t] {$w$};
  
 \path[sampleEdge,<-]
    (t) edge [left] node [above, sloped] {$-0.2$}  (v);
 \path[sampleEdge,<-]
    (v) edge [left] node [above,sloped] {$-0.2$} (s) ;
  \path[sampleEdge,<-]
    (t) edge [left] node [above, sloped] {$-0.2$}  (w);
  \path[sampleEdge,<-]
    (w) edge [left] node [above, sloped] {$-0.2$}  (s);
 
\end{tikzpicture}

&

\begin{tikzpicture}[->,auto,node distance=2cm,
  thick,main node/.style={circle,draw,font=\footnotesize\sffamily}]
  \tikzstyle{sample} = [circle,draw,font=\footnotesize\sffamily,minimum size=2.5em]
  \tikzstyle{sampleEdge} = [font=\sffamily\small]
	
   \node[sample] (t) {$t$};
  \node[sample] (v) [above left of=t] {$v$};
  \node[sample] (s) [below left of=v] {$s$};
  \node[sample] (w) [below left of=t] {$w$};
  
 \path[sampleEdge,<-]
    (t) edge [left] node [above, sloped] {$-0.2$}  (v);
 \path[sampleEdge,<-]
    (v) edge [left] node [above,sloped] {$-0.2$} (s) ;
  \path[sampleEdge,<-]
    (t) edge [left] node [above, sloped] {$-0.2$}  (w);
  \path[sampleEdge,<-]
    (w) edge [left] node [above, sloped] {$-0.2$}  (s);
  \path[sampleEdge,<-]
    (t) edge [left] node [above, sloped] {$-0.1$}  (s);
   
\end{tikzpicture}

\\

(a) & (b)

\end{tabular}
\caption{Two graphs with partial correlations on the edges.  In (a) applying criterion (\ref{eq:tr-criterion}) to the graph leads to the correct conclusion that $s$ is not a parent of $t$. We obtain for the correlation $\rho_{ts}=2(-0.2)(-0.2)=0.08<0.2=\min |\rho_{ij}|$, where the minimum is obtained from the two directed paths from $s$ to $t$. In (b) the incorrect inference is obtained. We have the correlation $\rho_{ts}=2(-0.2)(-0.2)+(-0.1)=-0.02$ and $|-0.02|<0.2=\min |\rho_{ij}|$.}\label{fig:wright-graph}
\end{figure}
Criterion (\ref{eq:tr-criterion}) does not hold in general because for smaller values of the coefficient $\beta_{ts}$ for $s\to t$, the correlation will not necessarily exceed the smallest value of the coefficients on the directed paths from $s$ to $t$ (see Lemma \ref{lem:path-prod} in Appendix \ref{sec:appendix-lemmas}).

That using criterion (\ref{eq:tr-criterion}) in transitive reduction cannot in general lead to the correct graph was confirmed by simulations in \citet{Kossakowski:2021}, where the TRANSWESD and DR-FFL algorithms were implemented. The simulations showed that the true underlying graphs were poorly reconstructed. Another way to deal with alternative paths is to condition on other nodes on directed paths between the nodes of interest, as in the classical way of \citet{Pearl:2001} and \citet{Spirtes:1993}.  This is the suggestion of \citet{Peters:2015} and \citet{Meinshausen:2016} and \citet{Magliacane:2016}, also described in \citet{Mooij:2020a}, which we describe next.

\subsection{Invariant causal prediction}\label{sec:cond-inv-pred}\noindent
Following the ideas of \citet{Pearl:1988,Pearl:2000} and \citet{Spirtes:1993}, determining direct causes (parents) $s\to t$ can be done by conditioning on intermediate variables. Consider the graph $s\to u\to t$ (part of the graph in Figure \ref{fig:perturb-graph}(a)). We find a correlation between $s$ and $t$, but conditioning on $u$ results in a partial correlation of 0, no matter if we intervene on $s$ or not. This is the seminal idea of \citet{Peters:2015} and \citet{Meinshausen:2016} and \citet{Magliacane:2016}, where invariant prediction was used as a guiding principle, following \citet{Pearl:2000} and \citet{Eberhardt:2007}, for instance. 

Consider the graph in Figure \ref{fig:perturb-graph}(a), where we intervene on node $s$. Marginal invariant prediction would use observational data and data from the intervention on $s$, and come to the conclusion that $\beta_{ts}\ne 0$ in both contexts, and so $s$ is relevant for predicting $t$ but $s$ may not be a parent of $t$. 
If we include node $u$ in the set of predictors for $t$, i.e., we use $(X_{s},X_{u})$ as predictors for $X_t$, then $u$ will block the path from $s$ to $t$ ($s$ and $t$ are $d$-separated by $u$). This is illustrated in Figure \ref{fig:scatterplots-invariant-prediction}, where in Figure \ref{fig:scatterplots-invariant-prediction}(a) no conditioning takes place and we find evidence that there is a directed path $s\to\cdots \to t$, and in Figure \ref{fig:scatterplots-invariant-prediction}(b) $u$ is conditioned on and we find evidence that $u$ is in between $s$ and $t$ (see calculations in Appendix \ref{sec:appendix-calculations-example}). Consequently, we will find that the distribution of $X_{t}\mid x_{u}\mid\mid C=\emptyset$ is equal to the distribution of $X_{t}\mid x_{u}\mid\mid C=s$. This is equivalent to the residual distributions $X_s-X_u$ and $X_t-X_u$ being equal, as shown in Figure \ref{fig:scatterplots-invariant-prediction}(b). The conclusion will then be correct that $s$ is not a direct cause of $t$. 


%
\begin{figure}[t]\centering
\begin{tabular}{c @{\hspace{5em}} c}

\begin{tikzpicture}[->,auto,node distance=2cm,
  thick,main node/.style={circle,draw,font=\footnotesize\sffamily}]
  \tikzstyle{sample} = [circle,draw,font=\footnotesize\sffamily,minimum size=2.5em]
  \tikzstyle{sampleEdge} = [font=\sffamily\small]
	
  \node[sample] (t) {$t$};
    \node[sample] (s) [above left of=t] {$s$};
    \node[sample] (u) [below left of=s] {$u$};
  \node[sample] (v) [above left of=u] {$v$};
 \path[sampleEdge,<-]
    (t) edge node [left] {} (u); 
 \path[sampleEdge,<-]
    (u) edge node [left] {} (s) ; 
 \path[sampleEdge,<-]
    (s) edge node [left] {} (v) ; 
\path[sampleEdge,->,dashed]
    (s) edge node [left] {} (t) ; 
\end{tikzpicture}

&

\begin{tikzpicture}[->,auto,node distance=2cm,
  thick,main node/.style={circle,draw,font=\footnotesize\sffamily}]
  \tikzstyle{sample} = [circle,draw,font=\footnotesize\sffamily,minimum size=2.5em]
  \tikzstyle{sampleEdge} = [font=\sffamily\small]
	
  \node[sample] (t) {$t$};
    \node[sample] (s) [above left of=t] {$s$};
    \node[sample] (u) [below left of=s] {$u$};
  \node[sample] (v) [above left of=u] {$v$};
 \path[sampleEdge,<-]
    (t) edge node [left] {} (u); 
 \path[sampleEdge,<-]
    (t) edge node [left] {} (s) ; 
 \path[sampleEdge,<-]
    (s) edge node [left] {} (v) ; 
\end{tikzpicture}

\\

(a) & (b)

\end{tabular}
\caption{Two perturbation graphs with an intervention on node $s$ denoted by $\{s\}$. In (a) perturbing on $s$ will not result in a change in the distribution of $t$ if node $u$ is conditioned on. But, in a perturbation graph (not conditioning on $u$), then the dashed edge $s\to t$ obtains. In (b) an intervention on node $s$ will result in a change in the distribution of node $t$, no matter if node $u$ is conditioned on or not.   }
\label{fig:perturb-graph}
\end{figure}
\begin{figure}[t]\centering
\pgfimage[width=.45\textwidth]{sut-chain-same-var.pdf}
~
\pgfimage[width=.45\textwidth]{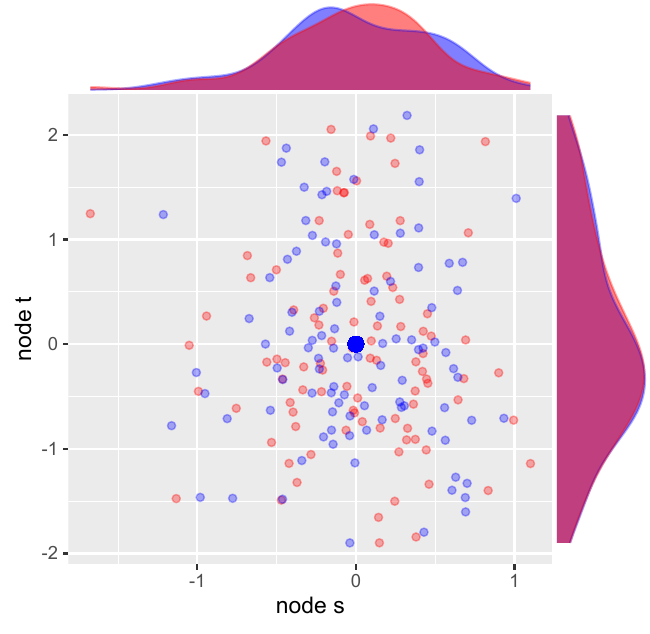}\\[-.5em]
(a)  \hspace{14em} (b)
\caption{ 
Scatterplots for $X_{s}$ and $X_{t}$ in (a) and the residuals $X_s-X_u$ and $X_t-X_u$ in (b) for the underlying graph $s\to u\to t$ in different contexts, where ${\color{red} C=\emptyset}$ and ${\color{blue} C=s}$ and the intervention on $s$ replaces $X_s$ by $W=2+ N(0,1)$. The larger filled circles represent the corresponding empirical means. In (a) the intervention on $s$ and no conditioning on $u$, which results in shifted distributions with respect to the observational distributions. In (b) the residuals $X_s-X_u$ and $X_t-X_u$ (conditioning on $u$), resulting in a 0 correlation (approximately) and no change in distribution (approximately).}
\label{fig:scatterplots-invariant-prediction}
\end{figure}

This example shows that indirect connections are removed by conditioning, both in the observation and intervention contexts. 
These considerations lead to the following definition \citep{Peters:2015}, extending the marginal invariance defined previously in Definition \ref{def:marginal-invariant-prediction} for the perturbation graph.
\begin{definition}{\rm (Conditional invariant prediction)} \label{def:conditional-invariant-prediction} {Let $(X_{1},X_{2},\ldots,X_{m-1})$ with values in $\mathbb{R}^{m-1}$ be the predictors of $X_t$, and let $\beta_{ts}$ be the coefficients for the linear regression of $X_t$ on $X_s$ for all $s \ne t$, as in model (\ref{eq:true-model}). The linear prediction is called {\em conditional invariant predictive for $t$} if there is a set $S\subseteq V\backslash\{t\}$ such that
\begin{align}\label{eq:conditional-invariant-prediction}
X_{t}=\sum_{s\in S}X_{s}\beta_{ts}+ \varepsilon_{t}\quad \text{for contexts}\quad C=\emptyset \text{ and } C=s
\end{align}
where $\varepsilon_{t}$ is Gaussian with mean 0 and variance $\sigma^{2}_t$ and $\varepsilon_{t}\independent X_{s}$. }
\end{definition}

Using Definition \ref{def:conditional-invariant-prediction} ensures that we learn about direct causes, and not only about causal paths, as is the case with a perturbation graph. However, there may be several sets of variables that lead to conditional invariant prediction. Consider Figure \ref{fig:perturb-graph}(a) again. If we take as conditioning variable set $S_{u}=\{u\}$, then intervening on $v$ will not lead to a change in $t$. Similarly, if we condition on $S_{s,u}=\{s,u\}$, then intervening on $v$ will not lead to a change in $t$. Hence there are two sets, $S_u$ and $S_{s,u}$ that satisfy conditional invariant prediction,  i.e., the conditional distributions $X_t\mid x_u\mid\mid C=v$ and $X_t\mid x_u,x_s\mid\mid C=v$ are equal. This implies that we must choose which set leads to the correct set of direct causes. \citet{Peters:2015} show that taking the common nodes among the sets that satisfy conditional invariant prediction (in our example $S_u$ and $S_{s,u}$), obtains the direct causes. In the example, we see that the intersection $S_u\cap S_{s,u}$ (common nodes of $S_u$ and $S_{s,u}$) gives us the correct set $S_u$, and so $u$ is a direct cause of $t$. Taking the common variables that give conditional invariant prediction in Definition \ref{def:conditional-invariant-prediction} thus leads to the unique and correct set of direct causes. 

By taking the intersection, the estimate of the direct causes is relatively small (conservative) and is, in fact, with high probability contained in the set of true direct causes \citep[][Theorem 1]{Peters:2015}. This follows from two things: (1) the set of common nodes (intersection) for which the conditional invariant prediction holds is a subset of the set of true direct causes, and (2) each test on the conditional invariant prediction for a set $S$ is of level $\alpha$, the significance level. Then the probability of the intersection being in the set of true direct causes is at least $1-\alpha$. To see this, we go to the example above related to Figure \ref{fig:perturb-graph}(a) again. We have two sets $S_u$ and $S_{s,u}$ and their intersection belongs to the set of true direct causes of $t$, $S_u$ (by (1)). And the probability for each not to reject is at least $1-\alpha$ if we test both at level $\alpha$. Then the probability of the intersection being in the set of true direct causes of $t$ is at least as large as the probability of the set of true direct causes not being rejected (because the intersection is smaller than the true set), which is at least $1-\alpha$ (by (2)). \citet[][Theorem 1]{Peters:2015} show by similar reasoning that from this it also follows that we can obtain accurate confidence intervals (good coverage). We illustrate these confidence intervals in Section \ref{sec:application}.

Invariant causal prediction solves the issue with transitive reduction, and \citet{Peters:2015} prove that invariant causal prediction is consistent for the correct subset of direct causal relations \citep[see also][]{Mooij:2020a}.
Sufficient conditions to obtain all parents (direct causes) of a node in the linear Gaussian model (\ref{eq:true-model}) have been obtained for hard interventions: (a) each of the predictor variables $X_s$ is intervened on, and (b) the means of the intervened variable is not the same as the mean in the observed context \citep[][Theorem 2(i)]{Peters:2015}. Although this is sufficient to obtain all parents, it is not necessary. For example, in the graph $t\leftarrow u\to s\to t\to v$ the correct set of parents of $t$ is $\{u,s\}$ and can be obtained by only intervening on $u$, an intervention on $s$ is not necessary \citep[][Appendix A]{Peters:2015}. There are likely many more situations where not all and only some nodes need to be intervened on, but currently not all necessary conditions are known; we briefly come back to this in the Discussion \ref{sec:conclusion-discussion}.



\subsection{Hidden or unobserved confounders}\label{sec:unobserved}\noindent
In causal analyses it is often assumed that all relevant variables are observed and included in the analysis (causal sufficiency). If causal sufficiency holds, then the assumption that the residual $\varepsilon_{t}$ is independent of the predictors $X_{s}$ for some $s\in S$ should hold. However, to assume causal sufficiency is rather presumptuous \citep{Pearl:2000}. For example, suppose that we have the graph of Figure \ref{fig:hidden-graph}(a). Here node $u$ is unobserved (hidden) and so we work with the marginal distribution over the nodes $s,t$ and a binary variable $C$ that induces a soft intervention such that $\expt(X_{s}\mid\mid C=0)\ne \expt(X_{s}\mid\mid C=1)$. 
\begin{figure}[t]\centering
\begin{tabular}{c @{\hspace{5em}} c}

\begin{tikzpicture}[->,auto,node distance=2cm,
  thick,main node/.style={circle,draw,font=\footnotesize\sffamily}]
  \tikzstyle{sample} = [circle,draw,font=\footnotesize\sffamily,minimum size=2.5em]
  \tikzstyle{sampleEdge} = [font=\sffamily\small]
	
  \node[sample] (t) {$t$};
    \node[sample] (u) [above left of=t] {$u$};
    \node[sample] (s) [below left of=u] {$s$};
  \node[sample] (C) [above left of=s] {$C$};
 \path[sampleEdge,<-]
    (t) edge node [left] {} (u) ;
 \path[sampleEdge,<-,dotted]
    (s) edge node [left] {} (u) ; 
 \path[sampleEdge,<-]
    (s) edge node [left] {} (C) ; 
 \path[sampleEdge,<-]
    (t) edge node [left] {} (s) ; 
\end{tikzpicture}

&

\begin{tikzpicture}[->,auto,node distance=2cm,
  thick,main node/.style={circle,draw,font=\footnotesize\sffamily}]
  \tikzstyle{sample} = [circle,draw,font=\footnotesize\sffamily,minimum size=2.5em]
  \tikzstyle{sampleEdge} = [font=\sffamily\small]
	
  \node[sample] (t) {$t$};
    \node[sample] (u) [above left of=t] {$u$};
    \node[sample] (s) [below left of=u] {$s$};
  \node[sample] (C) [above left of=s] {$C$};
 \path[sampleEdge,<-]
    (t) edge node [left] {} (u) ;
 \path[sampleEdge,<-]
    (s) edge node [left] {} (u) ; 
 \path[sampleEdge,<-]
    (s) edge node [left] {} (C) ; 
 \path[sampleEdge,<-]
    (t) edge node [left] {} (s) ; 
 \path[sampleEdge,<-,dotted]
    (u) edge node [left] {} (C) ;     
 \path[sampleEdge,<-,dotted]
    (s) edge [bend right] node [left] {} (t) ;     
\end{tikzpicture}

\\

(a) & (b)

\end{tabular}
\caption{Two graphs where $u$ is unobserved and $C$ is an intervention node. In (a) a graph where the intervention node $C$ removes the effect of node $u$ onto $s$. In (b) The intervention node violates (2) affecting both $s$ and $u$, and the edge from $t$ to $s$ violates (3) which states that there is no feedback.  }\label{fig:hidden-graph}
\end{figure}
If we know that node $u$ exists we can take $u$ into account and invariant causal prediction works as usual. But if we are unaware of node $u$, and we do not take it into account, then we can no longer assume that the residual $\varepsilon_{t}=X_{t}-X_{ts}\beta_{ts}$ and its support $(X_{s},X_{u})$ are independent (see Appendix \ref{sec:appendix-unobserved}). Hence, the effect of node $u$ remains and the assumption of independence of $\varepsilon_t$ and $X_s$ for all $s\in S$ in (\ref{eq:conditional-invariant-prediction}) is no longer valid. 

We can then take two different routes to infer causal relations in the setting of unobserved confounders. (i) We can weaken Definition  \ref{def:conditional-invariant-prediction} and allow for dependence between the residual and predictors; or (ii) we can include an additional variable for which we know the relation to the target and source variable (instrumental variable).

To begin with (i) where we weaken Definition \ref{def:conditional-invariant-prediction}, the idea is that this weakened definition represents causal effects that may have been caused by unobserved variables \citep[][Proposition 5]{Peters:2015}, and so the conclusion is weakened to ancestors as possible causes. The representation that is obtained from the measured variables is simply an incomplete picture but it is known that some predictors are ancestors of the target variable. 

The second route is to use an observed variable to remove the correlation between the residual and the predictors.  We therefore have to invoke an experiment with known causal relation, such that (1) for any soft univariate intervention from $C$ to $s$ the effect cannot be constant, e.g., $\expt(X_{s}\mid\mid C=0)\ne \expt(X_{s}\mid\mid C=1)$, (2) the variable $C$ ($=0$ or $1$) is connected to a single predictor, and (3) there is no feedback from the target variable $X_{t}$ to the other nodes \citep{Peters:2015}. Figure \ref{fig:hidden-graph}(b) shows two (dotted) edges that violate the criteria. The dotted edge $(t,s)$ violates (3) that there be no feedback from the target variable, and the edge $(C,u)$ violates (2) that the instrumental variable only affects a single node. In practice, we can choose an intervention on a particular (single) node $s$ as the instrumental variable $C$. Considering the assumptions (1)-(3) we must demand that the intervention has no feedback from the target node to any of the other nodes (criterion (3)). So, this is basically a stronger version of invariant causal prediction, where we assume a more fine-grained intervention of the environment (the experiment nudges a particular variable), and an assumption on no feedback. Especially condition (2), that the intervention is specific to only a single node, is problematic in the social sciences.

\section{Application to psychology data}\label{sec:application}\noindent
In psychology controlled experiments are considered a valuable tool to infer causal relations. 
The causal conclusions from such experiments are similar to perturbation graphs in that the factors are ancestors of the changes in distributions (often the means) of the dependent variable. 
Invariant causal prediction can be considered as a generalisation of more traditional experiments in psychology because: (1) using additional variables (covariates) leads to more information on direct causes (parents), and (2) different changes in distribution (other than the mean) are considered to assess causal effects. Here we discuss two example data sets; one which explicitly involved the methodology of perturbation graphs by intervening on multiple variables one by one \citep{Hoekstra:2018}, and another dataset where we have a single time series of a person diagnosed with major depressive disorder \citep{Wichers:2016}.  

\subsection{Meat eating example}\noindent
The first dataset we use involves the topic of meat consumption and is fully described in \citet{Hoekstra:2018}. The design resembles the graph perturbation method of \hyperref[sec:cond-cor]{Section \ref*{sec:cond-cor}}. First, a questionnaire about attitudes toward meat consumption (e.g., ``The production of meat is harmful for the environment'') is administered without any intervention (observation or wild type). Subsequently, a hypothetical scenario is presented (e.g., ``The meat and dairy industry has a huge CO-2 emission and is therefore harmful for the environment. How does this influence your attitude towards the consumption of meat?'') at which point the same questionnaire is administered again. This is repeated for each of the 11 statements in the questionnaire (see the caption of Figure \ref{fig:meat-graph}). The responses were measured on a 7 point Likert scale.

\begin{figure}[t]\centering
\begin{tabular}{c @{\hspace{2em}} c}

\scalebox{0.8}{\begin{tikzpicture}[->,auto,node distance=3cm, 
  thick,main node/.style={circle,draw,font=\footnotesize\sffamily}]
  \tikzstyle{sample} = [circle,draw,font=\footnotesize\sffamily,minimum size=3em]
  \tikzstyle{sampleEdge} = [font=\sffamily\small]
	
    \node[sample] (guilty) {guilty};
    \node[sample] (infer) [above left of=guilty] {infer};
    \node[sample] (suff) [below of=guilty] {suff};
    \node[sample] (tax) [right of=infer] {tax};
    \node[sample] (moral) [above right of=tax] {moral};	
    \node[sample] (taste) [above left of=tax] {taste};
    \node[sample] (nutr) [right of=guilty] {nutr};	    
    \node[sample] (death) [left of=guilty] {death};	
    \node[sample] (envir) [below right of=guilty] {envir};
    \node[sample] (sad) [below left of=guilty] {sad};	
    \node[sample] (disg) [right of=tax] {disg};	      

 \path[sampleEdge,->]
    (guilty) edge[bend right=20] node [left] {} (infer)
    (guilty) edge node [left] {} (death)
    (infer) edge[bend right=20] node [left] {} (guilty)    
    (suff) edge node [left] {} (guilty)
    (tax) edge node [left] {} (guilty)
    (taste) edge node [left] {} (guilty)    
    (moral) edge node [left] {} (tax) ; 

\end{tikzpicture}}
&
\pgfimage[width=.5\textwidth]{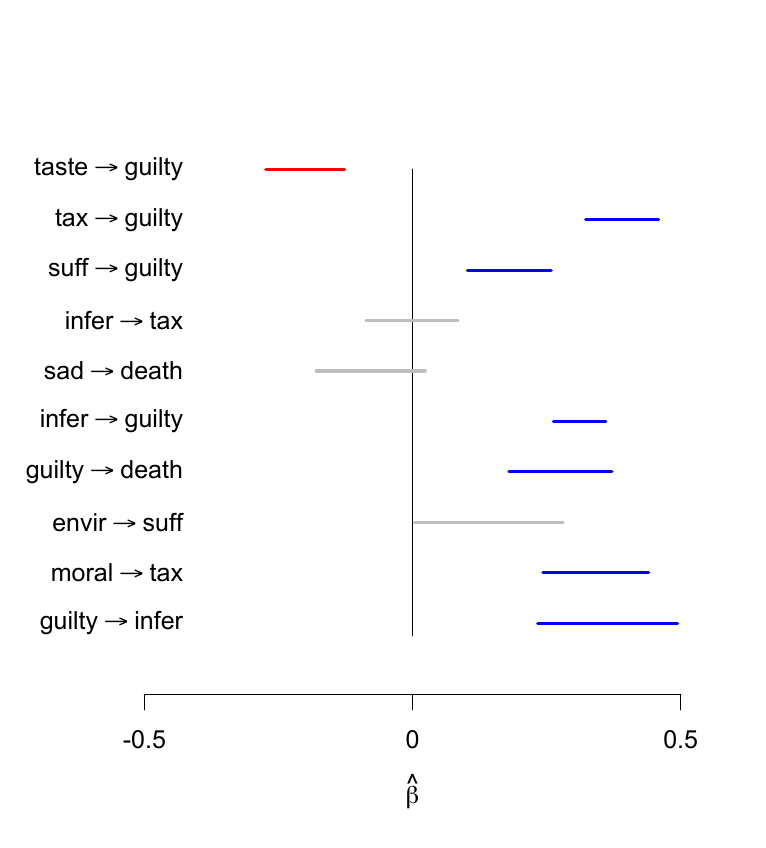}\\
(a) &(b)
\end{tabular}
\caption{(a) Graph of the conditional invariant prediction method applied to the data on attitudes of meat consumption. 
The nodes represent: \textsf{moral}: eating is morally wrong, \textsf{nutr}: meat contains important nutrients, \textsf{envir}: production of meat is harmful to the environment, \textsf{infer}: animals are inferior to people, \textsf{suff}: by consuming meat you contribute to animal suffereing, \textsf{tax}: there should be a tax on meat, \textsf{taste}: I like the taste of meat, \textsf{death}: meat reminds me of death and suffering of animals, \textsf{sad}: if I had to stop eating meat I would feel sad, \textsf{guilty}: if I eat meat I feel guilty, \textsf{disg}: if I eat meat I feel disgust. (b) $95\%$ Confidence intervals of 10 different edges. Blue lines indicate significant (and are included in the graph in (a)), red lines indicate negative values, and gray indicates not significant (there are 45 more confidence intervals that contain 0 which are not shown). 
 }
\label{fig:meat-graph}
\end{figure}

For the analysis the conditional invariant prediction algorithm is applied to each node separately.  In \citet{Kossakowski:2021} we also analysed these data but there we used an idea similar to that used in \citet{Meinshausen:2016}, where each pair of contexts was used to obtain an estimate of causal relations. Here we apply a more straightforward approach where simply all contexts (perturbations) were given as input. We assumed tacitly (after checking that the distributions were approximately unimodal and symmetric) that the Likert scales could be treated as continuous random variables. The \textsf{R} package \texttt{\small InvariantCausalPrediction} is used. The algorithm produces for each of the contexts regression confidence intervals and chooses the lower bound of the confidence interval (minimax estimate). Invariant causal prediction obtains an appropriate model for each context, and as such assumes that the residual distributions that are compared among the contexts are independent (which the theory prescribes is true). We set the level to 0.05 for each context (1 observation and 11 perturbation contexts). The test for the conditional distributions was the Kolmogorov-Smirnov test. The Kolmogorov-Smirnov test takes the largest absolute discrepancy of two cumulative probability distributions as a test statistic, and can be shown to be distribution-free \citep[see, e.g., ][Chapter 19]{Vandervaart98}. 

In Figure \ref{fig:meat-graph}(a) we see that seven out of the 11 nodes are connected with an edge. There is a chain from \textsf{moral} $\to$ \textsf{tax} $\to$ \textsf{guilty}, and \textsf{suff}, \textsf{taste} and \textsf{infer} also have an arrow into \textsf{guilty}. Interestingly, the graph shows a feedback loop \textsf{infer} $\rightleftarrows$ \textsf{guilty}, which is a cycle, assuming bounds on the weights \citep[see][for details, the absolute eigenvalues of the coefficient matrix all $<1$]{Rothenhausler}.

Figure \ref{fig:meat-graph}(b) shows (some of) the $95\%$ confidence intervals of the edges (there are 45 confidence intervals not shown which contain 0). These confidence intervals are obtained by pooling the data over the contexts, given that there is evidence that the conditional distributions are the same across different contexts \citep{Peters:2015}. These confidence intervals can therefore be interpreted as usual. The edges in the graph in Figure \ref{fig:meat-graph}(a) are confirmed by the confidence intervals in Figure \ref{fig:meat-graph}(b).

\subsection{Time series of a patient with major depressive disorder}\noindent
Another way to apply the theory of invariant causal prediction is to time series \citep{Pfister:2019}. The basic idea is to split the time series into different consecutive parts that are considered different contexts. The reason a section (block) of a time series can be considered as a context is that there is a particular event or change that has a particular impact. The advantage of the invariant causal prediction approach is that the assumptions about the interventions are not strong.  It is assumed that there is an effect on the distribution, but there could be multiple variables involved. 

Here we apply invariant causal prediction to a time series of a patient diagnosed with major depressive disorder. The time series, obtained with ecological momentary assessment, consists of 1478 measurements over the course of 239 consecutive days \citep{Kossakowski:2017}. During the measurements the antidepressants were weaned off, where neither the patient nor the physician knew at which times exactly the reduction in antidepressants started. In total we have four variables each consisting (sum score) of several items related to mood, physical condition, self-esteem and symptoms of depression.

We applied the algorithm for invariant causal prediction called \texttt{seqICP} \citep{Pfister:2019}. Testing for similarity of environments is done with an adaptation of the Chow test, which considers the normalised variance of the residuals after regression. We selected concurrent and previous timepoints for the linear regressions in predicting the variables. We set the significance level for the tests to 0.05. 

In Figure \ref{fig:groot-graph}(a) we see the resulting graph with four edges and in Figure \ref{fig:groot-graph}(b) the corresponding confidence intervals (the confidence intervals which include 0, and so are not significant, are not shown). We see that \textsf{mood} is not connected. This could be because \textsf{mood} was affected directly and so no set of direct causes could be established. Interestingly, \textsf{phys} (physical condition) and \textsf{symp} (symptoms) are in a feedback system, so that physical condition has an effect on symptoms and vice versa. This may be relevant for this patient to consider for interventions; it may be easier to improve physical condition directly than symptoms. 
\begin{figure}[t]\centering
\begin{tabular}{c @{\hspace{2em}} c}
\pgfimage[width=.5\textwidth]{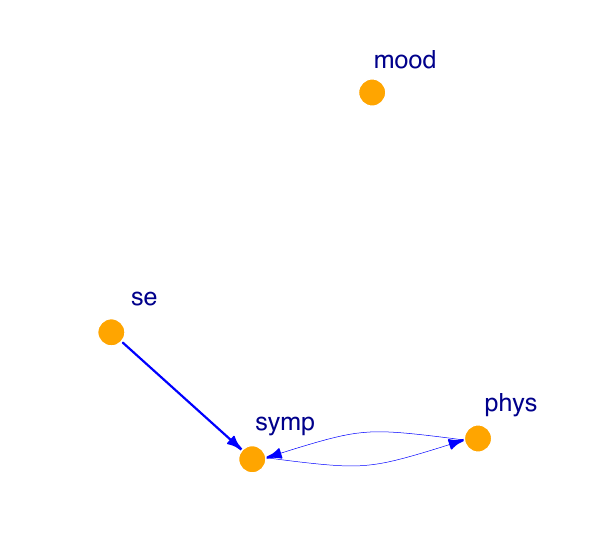}
&
\pgfimage[width=.5\textwidth]{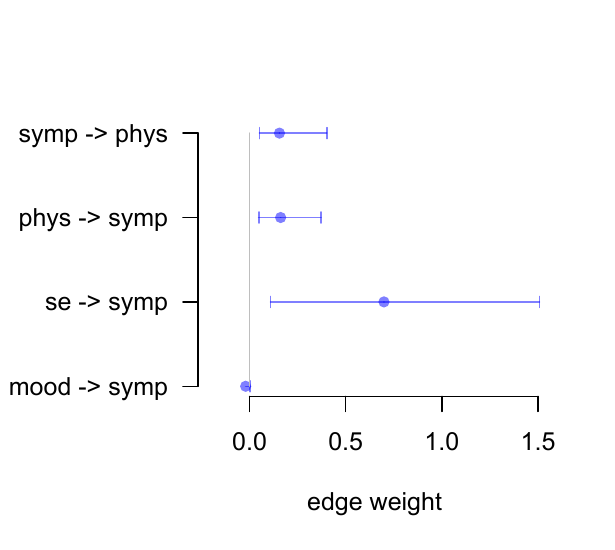}\\
(a) &(b)
\end{tabular}
\caption{(a) Graph of the time-series data of a single individual. The variables are each sum scores of several items: \textsf{mood} has 12 items, \textsf{symp} (symptoms) has four items, \textsf{se} (self esteem) has four items, and \textsf{phys} (physical condition) has five items. (b) $95\%$ Confidence intervals of the four significant edges. }
\label{fig:groot-graph}
\end{figure}
%

\section{Conclusion and discussion}\label{sec:conclusion-discussion}\noindent
We have considered two seemingly different ways of combining different data sets in a single analysis to obtain a causal graph: The perturbation graph method and invariant causal prediction. Invariant causal prediction can be considered as a generalisation of perturbation graphs in that perturbation graphs obtain ancestors (which is equivalent to marginal invariant prediction) while invariant causal prediction obtains a set of direct causes. 

Direct application of the two-step perturbation method (making a perturbation graph and then pruning superfluous connections) has no statistical guarantees of a correct solution. For making the perturbation graph, we have shown that the conditional correlation requires separate contexts to decide whether there is a directed causal path. Furthermore, we have shown that transitive reduction does not in general prune edges such that we obtain the correct graph. Using ideas from invariant causal prediction does lead to consistent estimates of the underlying graph. But the perturbation graph method does have merit in terms of the design of experiments in psychology. And in combination with invariant causal prediction, it is possible to discover or confirm hypotheses particular pathways between variables. 

We can view the design of the perturbation method, where some of the variables are intervened on consecutively to determine its effect on other variables, as an extension of the traditional experimental designs used in psychology. The effect of the intervened variable on  the distribution of the other variables can be any kind of change (e.g., mean, variance, shape, etc.). Also, the method is multivariate in nature, as the effects of the intervention are registered at all variables simultaneously. The design is appropriate for different kinds of experiments in psychology. There is a striking resemblence with some areas of brain research, where functional magnetic resonance imaging is used. Change in one brain area, possibly due to some task, is considered a source (ancestor) of other brain areas. Correlational measures are then used to determine which brain areas are effected by the source \citep{Roebroeck:2009}. Applying invariant causal prediction will benefit such results possibly identifying direct causes instead of just ancestors. 

The perturbation graph in combination with invariant causal prediction as described and applied above also has limitations; we discuss two. First, in the attitude on meat-eating experiment we tacitly assumed that each variable was intervened on in succession by the scenarios presented to the participants. Of course we have no guarantee that with each scenario exactly one variable was intervened upon; we could have intervened on several simultaneously \citep[fat-handed interventions, see e.g.,][]{Eberhardt:2007, Eronen:2020}. Although not ideal, it is still possible to obtain some causal relations. In soft interventions (where the variables are not completely under control) the different interventions could be pooled (it can be tested whether this is appropriate) and compared with the observation data; it is then still possible to determine several causal relations \citep[][Section 4.3]{Peters:2015}. 

Second, in the meat eating example we intervened on \textit{all} variables in succession and measured the effects. In many settings it is likely that at least one of the variables cannot be intervened on. For instance, if a patient diagnosed with major depression disorder ruminates before bedtime and therefore sleeps poorly, we would have great difficulty to intervene directly on rumination; and so we would not be able to determine the effects if we could change rumination. It turns out that not all experiments on all variables are required to obtain the correct causal graph. If we had three variables, then we need to intervene on only two variables; more generally, \citet[Theorem 3.3]{Eberhardt:2005} prove that for $m$ variables we need at most $\log(m)+1$ experiments. For example, with $10$ variables, we need to intervene on at most four variables. So, even though in general we will not be able to intervene on all variables, this is not necessary to obtain a correct causal graph.

All in all, the combination of perturbation graphs and invariant causal prediction seems to be a promising route in discovering and possibly testing (confirming) not only causal relations in psychology but also causal pathways or mechanisms.


\appendix

\section{Appendix: Graphical models and assumptions}\label{sec:appendix-gm-assumptions}\noindent
We use the causal Markov condition to connect graphs and probability distributions using the relation between $d$-separation and conditional independence, as described in \citet{Pearl:2000}, \citet{Pearl:1992}, \citet{Spirtes:1993}, \citet{Peters:2017} and \citet{Maathuis:2018}. Using this allows us to identify the nodes in graph with the random variables $X=(X_1,X_2,\ldots,X_m)$, which together constitute a graphical model. The edges in a graphical model are defined by an association, like partial correlations for multivariate normal variables. 

The $d$-separation relation in graph $\mathcal{G}$ corresponds to a probability distribution $\pr$ (with density $p$ if it exists) via the Markov condition.  Let $X_s$, with $s=1,2,\ldots,m$, be random variables corresponding to the nodes of the graph in $V$. We denote by $X=(X_1,X_2,\ldots, X_m)$ an $m$-dimensional vector of random variables corresponding to the $m$ nodes with density $p$ (we use $p$ to denote the joint as well as the marginal densities). A subset of random variables is denoted by $X_A$, i.e. the set of variables $(X_s, s\in A)$. Random variables $X_s$ and $X_t$ are independent if, for density $p$, for all $x_s$ and $x_t$, $p(x_s,x_t)=p(x_s)p(x_t)$; and the random variables $X_s$ and $X_t$ are conditionally independent given $X_u$ if $p(x_s,x_t\mid x_u)=p(x_s\mid x_u)p(x_t\mid x_u)$ for all $x_s$, $x_t$ and $x_u$. For disjoint sets of nodes $A$, $B$ and $C$, we have the random vectors $X_{A}$, $X_{B}$ and $X_{C}$, and we denote that $X_{A}$ is conditionally independent of $X_{B}$ given $X_{C}$ by $X_{A}\independent X_{B}\mid X_{C}$ \citep{Dawid79}, meaning that $p(x_A,x_B \mid x_C)=p(x_A\mid x_C)p(x_B\mid x_C)$, where $x_A$ are the $x_i$ such that $i\in A$.Then the assumption that connects a graph to a probability distribution is the causal Markov condition. 

\begin{assumption}{\em (Causal Markov condition)}\label{ass:markov}
We assume that for all disjoint sets of variables $A$, $B$, and $C$ the {\em causal Markov condition} is satisfied, which specifies that 
\begin{align}
A\independent B\mid C \quad\implies\quad X_{A} \independent X_{B}\mid X_{C}
\end{align}
\end{assumption}
For example, the graph in Figure \ref{fig:d-separation-example}(a) has the $d$-separations $u\independent v$ and $\{u,v\}\independent t\mid s$. By the causal Markov condition this implies that any distribution $\mathbb{P}$ compatible with this graph must have the conditional independencies $X_u\independent X_v$ and $\{X_u,X_v\} \independent X_t\mid X_s$.

Unfortunately, the causal Markov condition does not always lead to a unique distribution. If we consider the subgraph $u\to s\to t$ for the moment, we see that the Markov condition implies that for any distribution the conditional independence $X_u\independent X_t\mid X_s$ must hold. Then the Markov factorisation (see Appendix \ref{sec:appendix-gm-assumptions}) can be both $p(x_u)p(x_s\mid x_u)p(x_t\mid x_s)$ and $p(x_t)p(x_s\mid x_t)p(x_u\mid x_s)$. 
\citet[][Theorem 1]{Verma:1991} showed that such an equivalence can be characterised by having the same $d$-separations. For example, $s\to u\to t$ is Markov equivalent to $s\leftarrow u\leftarrow t$ and $s\leftarrow u\to t$; these structures have the same $d$-separation $s\independent t\mid u$. In contrast, the collider structure $s\to u\leftarrow t$ has the $d$-separation $s\independent t$ \citep{Dawid:2015}. This implies, unfortunately, that without additional information, it is impossible to obtain all directions of the edges \citep{Pearl:2000}. This is why it is relevant to include interventional (perturbation) data, as we will see.

The causal Markov condition \ref{ass:markov} is equivalent to a factorisation of the joint distribution (if a density exists) where only the parents of each node are involved.  Let $X=x$ be a particular instantiation of the $p$ random variables in $X$. Then the \textit{Markov factorisation} with respect to a graph states that the distribution $p(x)$ can be represented by the product of conditional densities $p(x_s\mid x_{\textsc{pa}(s)})$, using only the variables that correspond to the parents of each of the nodes
\begin{align}\label{eq:markov-facorisation}
p(x) = p(x_{1}\mid x_{\textsc{pa}(1)})p(x_{2}\mid x_{\textsc{pa}(2)})\cdots p(x_{n}\mid x_{\textsc{pa}(m)})
\end{align}
\citep[For a proof see][Theorem 3.27.]{Lauritzen96} This factorisation is extremely useful, not only because of computational convenience, but also to define the intervention distribution (see \hyperref[sec:appendix-interventions]{Appendix \ref*{sec:appendix-interventions}}). To illustrate the Markov factorisation, consider the graph in Figure \ref{fig:d-separation-example}(a). By the factorisation we can determine the probability of each $X_s$ by considering $X_{\textsc{pa}(s)}$.  The nodes $u$ and $v$ have no parents, and so these probabilities will be a function of $X_u$ and $X_v$ alone. The parents of $s$ are $u$ and $v$, and so the probability of $X_s$ will be conditional on $X_u$ and $X_v$. And, finally, the probability of $X_t$ depends on $X_s$, since $s$ is a parent of $t$. Then the graph in Figure \ref{fig:d-separation-example}(a) can be represented as
\begin{align*}
p(x_v)p(x_u)p(x_s\mid x_u,x_v)p(x_t\mid x_s)
\end{align*}

The causal Markov condition in \ref{ass:markov} allows us to conclude that $d$-separations imply conditional independencies, but also the contrapositive, that if we find conditional dependencies, then this implies $d$-connectedness. However, the causal Markov condition does \textit{not} allow us to infer:  conditional independence implies $d$-separation \citep[][Chapter 18]{Maathuis:2018}. Hence, we cannot obtain $d$-separations by considering the distribution. Therefore, we require that all conditional independencies are covered by the causal Markov condition. This assumption is known as the {\em causal faithfulness condition}.
\begin{assumption}{\em (Causal faithfulness condition)}\label{ass:faith}
We assume that for disjoint sets of variables $A$, $B$, and $C$ the {\em causal faithfulness condition} is satisfied, which specifies that 
\begin{align}
A \independent B\mid C \quad\Longleftarrow\quad X_{A}\independent X_{B}\mid X_{C}
\end{align}
\end{assumption}
The faithfulness condition is violated more easily than the Markov condition. For instance, in a linear model for $s\to u\to t$ and $s\to t$ such that $\beta_{us}\beta_{tu}=\beta_{ts}$where $\beta_{ij}$ is the regression coefficient from $X_j$ to $X_i$ (see below), then the faithfulness condition is violated because $X_{s}\independent X_{t}$ but $s\not\independent t$ \citep{Pearl:2000,Buhlmann:2011,Peters:2017}.  Although the faithfulness condition can be violated, the occurrence of this violation has measure zero, as violation entails exact cancellation of effects (see Spirtes et al. \citeyear[][Theorem 3.2]{Spirtes:1993}, and B\"{u}hlmann and van de Geer \citeyear[][Section 13.6.1]{Buhlmann:2011}, but also \citet{Uhler:2013} for strong faithfulness). Faithfulness is sometimes referred to as stability \citep{Pearl:2000}. A slightly weaker assumption called minimality could also be used instead of faithfulness \citep[see][Proposition 6.36]{Peters:2017}, but we will not use this here. 


\section{Appendix: Interventions and context}\label{sec:appendix-interventions}\noindent
One of the key ideas for perturbation graphs and causal graphs is that it is possible to differentiate between conventional and interventional conditioning. Conventional conditioning refers to the idea that given specific information the probability should be modified, i.e., given subsets of nodes $A$ and $B$ the probability of $x_{A}$ given $x_{B}$ is \citep{Lauritzen:2001}
\begin{align*}
p(x_{A}\mid x_{B}) = \frac{p (x_{A}, x_{B})}{p(x_{B})}
\end{align*}
In contrast, interventional conditioning refers to the idea that there is some context where an experiment was conducted and an intervention was applied to alter the distribution. For instance, to improve math scores a teacher may decide to change teaching method, and thus may change (on average) the math score on a test in her class. Hence, the value for some random variable is changed as a result of the change of teaching method. 

We use the framework of \citet{Mooij:2020a} and augment the graph $\mathcal{G}=(V,E)$ with context nodes $K$ and edges such that we only allow $k\to i$ with $k\in K$ and $i\in V$, and not vice versa, i.e., $k\leftarrow i$ is not allowed. The augmented graph $\mathcal{G}^{a}=(V,K,E^{a})$, with $E^{a}=E\cup \{k\to i: k\in K, i\in V\}$. The nodes in $V$ are referred to as system nodes and the nodes in $K$ as context nodes. Associated with the context nodes are context variables $C_{k}$ for $k\in K$ with their distribution. We use $C_k=\emptyset$ to indicate that no intervention took place, but sometimes use $C$ for the context variable without subscript if the it is clear what context is referred to. 
\subsection{Hard interventions}\label{sec:hard-interventions}\noindent
Let $C$ be a random variable which indexes a specific experimental (or control) condition, as in Section \ref{sec:graphs}. The distribution $p(C=c)$ could be induced by experimental design. For example, to compare two methods of teaching math, method 1 or 2, a group of pupils could be randomly assigned to one of the two teaching methods obtaining $p(C=1)=p(C=2)=\tfrac{1}{2}$. Randomly assigning a pupil to method 1 or 2 is different from simply observing that a pupil is taught method 1 or 2, in the sense that other influences (confounders) like the parental guidance and stimulation to do math at home will be (more or less) equally distributed among the two groups. If there was no random assignment, then such influences would still be present. So, randomisation implies that the outcome cannot be a cause of assignment to a group (because assignment precedes the outcome) and there are no confounds between the assignment and outcome \citep{Sobel:2009}. Under these assumptions \citet[][Proposition 10]{Mooij:2020a} show that such a simple design with randomisation results in the correct causal inference that the teaching method causes the outcome (dependence between teaching method and outcome).

We have derived what is often referred to as a hard intervention, where the effects of other variables that affect the cause are removed so that all possible confounds have been eliminated. To denote interventional conditioning we use the notation $p(x\mid\mid C=c)$ \citep{Lauritzen:2002}. We write $C=\emptyset$ to indicate no intervention and so the observational distribution obtains \citep{Mooij:2020a}, i.e., $p(x\mid\mid C=\emptyset)=p(x\mid C=\emptyset)$. Since we are interested in interventions on specific nodes, we use $C=s$ to denote that we have applied an intervention to node $s$.

\begin{figure}[t]\centering
\begin{tabular}{c @{\hspace{5em}} c}

\begin{tikzpicture}[->,auto,node distance=2cm, 
  thick,main node/.style={circle,draw,font=\footnotesize\sffamily}]
  \tikzstyle{sample} = [circle,draw,font=\footnotesize\sffamily,minimum size=2.5em]
  \tikzstyle{sampleEdge} = [font=\sffamily\small]
	
  \node[sample] (t) {$X_{t}$};
  \node[sample] (u) [above left of=t] {$X_{u}$};
  \node[sample] (s) [below left of=u] {$X_{s}$};
  \node[sample,dashed] (c) [above left of=s] {$C$};

 \path[sampleEdge,<-]
    (t) edge node [left] {} (u);
 \path[sampleEdge,<-]
    (t) edge node [left] {} (s); 
 \path[sampleEdge,<-]
    (s) edge[] node [left] {} (u); 
 \path[sampleEdge,->,dashed]
    (c) edge[] node [left] {} (s); 
\end{tikzpicture}

&

\begin{tikzpicture}[->,auto,node distance=2cm,
  thick,main node/.style={circle,draw,font=\footnotesize\sffamily}]
  \tikzstyle{sample} = [circle,draw,font=\footnotesize\sffamily,minimum size=2.5em]
  \tikzstyle{sampleEdge} = [font=\sffamily\small]
	
  \node[sample] (t) {$X_{t}$};
  \node[sample] (u) [above left of=t] {$X_{u}$};
  \node[sample] (s) [below left of=u] {$X_{s}$};
    \node[sample] (c) [above left of=s] {$C$};
  
 \path[sampleEdge,<-]
    (t) edge node [left] {} (u);
 \path[sampleEdge,<-]
    (t) edge node [left] {} (s); 
 \path[sampleEdge,->,dashed]
    (u) edge node [left] {} (s) ; 
 \path[sampleEdge,->]
    (c) edge[] node [left] {} (s); 
\end{tikzpicture}

\\

\qquad $C=\emptyset$ & \qquad $C=1, 2$

\end{tabular}
\caption{(a) Hard intervention graph where the context variable is set to $C=\emptyset$ indicating no intervention, indicated by the dashed node and edge from $C$ to $X_{s}$, and so the observational distribution obtains. (b) Hard intervention where $C=1$ or 2 determines completely the influence on $X_{s}$ and so the effect of $X_{u}$ on $X_{s}$ is removed, indicated by the dashed edge. }\label{fig:intervention-graph}
\end{figure}

Consider Figure \ref{fig:intervention-graph} with the system variables $X_{u}$, $X_{s}$ and $X_{t}$ and the context variable $C$. The context variable $C$, in this example, can be $\emptyset$ (no intervention), and $1$ (e.g., group 1) or $2$ (e.g., group 2). Using the Markov factorisation in (\ref{eq:markov-facorisation}) we obtain the observational distribution when we set $C=\emptyset$, where nothing in the graph is changed 
\begin{align*}
p(x\mid\mid C=\emptyset)=p(x_{t}\mid x_{u},x_{s})p(x_{s}\mid x_{u}\mid\mid C=\emptyset)p(x_{u})
\end{align*}
Because we have $C=\emptyset$, we obtain 
\begin{align*}
p(x_{s}\mid x_{u}\mid\mid C=\emptyset)=p(x_{s}\mid x_{u}, C=\emptyset)\quad \text{so that} \quad p(x\mid\mid C=\emptyset)=p(x\mid C=\emptyset)
\end{align*}
In words, the observational distribution is the same as the conventional conditional distribution. This corresponds to Figure \ref{fig:intervention-graph}(a), where the presence of the edge $X_{u}\to X_{s}$ indicates no effect of $C$ on $X_{s}$. When $C=s$, however, we obtain the intervention distribution
\begin{align*}
p(x\mid\mid C=s)=p(x_{t}\mid x_{u},x_{s})p(x_{s}\mid x_{u}\mid\mid C=s)p(x_{u})
\end{align*}
Because $C=s$ and we intervene on $X_{s}$, the effect of $X_{u}$ on $X_{s}$ is removed, i.e., here we have that 
\begin{align*}
p(x_{s}\mid x_{u}\mid\mid C=s)=p(x_{s}\mid C=s)\quad\text{and so}\quad p(x\mid\mid C=s)\ne p(x\mid C=s)
\end{align*}
Hence, in the intervention distribution when $C=s$ we have that the parent of $s$, $\textsc{pa}(s)=\{u\}$, has been removed from the conditional probability. This corresponds to Figure \ref{fig:intervention-graph}(b), where the dashed edge $X_{u}\to X_{s}$ indicates that there is an intervention on $X_{s}$ from $C=s$ such that the effect of $X_{u}$ on $X_{s}$ has been removed. 
These considerations lead to the following definition. 
\begin{definition}{\em (Hard intervention)}\label{def:hard}
Let $\mathcal{G}^{a}$ be an augmented graph with system nodes $j$ in $V$ and context nodes $k$ in $K$ corresponding to context variable $C_k$, such that $k\to j$ only and not vice versa. A {\em hard intervention} where $C_{k}\ne \emptyset$, specifies that $p(x_{j}\mid x_{\textsc{pa}(j)}\mid\mid c_{k})=p(x_{j}\mid c_{k})$ such that there is some $k\in K$ with $k\to j$. All other terms in Markov factorisation of the joint distribution remain the same as without intervention. The \emph{hard intervention} distribution is
\begin{align}\label{eq:distribution-hard-intervention}
p(x\mid\mid c) = \prod_{j\in V\backslash \textsc{ch}(K)}p(x_{j}\mid x_{\textsc{pa}(j)})\prod_{i\in \textsc{ch}(K)}p(x_{i}\mid c_{\textsc{pa}(i)})
\end{align}
\end{definition}
Note that a hard intervention only specifies that the distribution of $X_{j}\mid\mid X_{\textsc{pa}(j)}, C_{k}$ is replaced by the distribution of $X_{j}\mid C_{k}$, but the variable is not replaced by a single value (almost surely with respect to $\pr$) as is common in, for instance, \citet{Pearl:2000} and \citet{Lauritzen:2001}. 
Such a hard intervention is also sometimes known as an \textit{ideal} \citep{Spirtes:1993} or a \textit{surgical} \citep{Pearl:2000} or a \textit{perfect} \citep{Mooij:2020a} intervention. We can think of our intervention as being not completely successful \citep{Mooij:2020a} or if there is an interaction between the intervention and observational unit such that small differences occur (noise, as is often assumed in the social sciences). 

The context variable is analogous to using dummy variables in statistics, in, for instance, analysis of variance (anova). In an anova the outcome is regressed on $C$ such that the regression coefficients result in the difference between two conditions (if $C$ is coded to 0 or 1) \citep{Agresti:1997}.

Using the linear Gaussian model (\ref{eq:true-model}) we can rewrite the model associated with Figure \ref{fig:intervention-graph}. We have for the observational setting $C=\emptyset$ the relations 
\begin{align*}
X_{u}&=\varepsilon_{u}
	&X_{s} &= \beta_{su}X_{u} + \varepsilon_{s}
	&X_{t} &= \beta_{ts}X_{s} + \beta_{tu}X_{u} + \varepsilon_{t}
\end{align*}
where the $\varepsilon_{i}$ are independently normally distributed with mean 0 and variance $\sigma^{2}$. We clearly have that $p(x_{t}\mid x_{s}\mid\mid C=\emptyset)=p(x_{t}\mid x_{s}\mid C=\emptyset)$, since there is no intervention.

When $C=s$ we intervene on $X_{s}$ by replacing {\color{red} it} by the random variable $W\in \mathbb{R}$ with mean $\mu_{w}\ne 0$ and variance $\sigma^{2}_{w}>0$ and obtain $X_{s} = W$. 
Then we find for the expectation of $X_{t}$
\begin{align*}
\expt(X_{t}\mid\mid C=s)=\beta_{ts}\mu_{w}\ne 0=\expt(X_{t}\mid\mid C=\emptyset)
\end{align*}
and for the variance 
\begin{align*}
\text{var}(X_{t}\mid\mid C=s)&=\beta_{ts}^{2}\sigma^{2}_{w}+(1+\beta_{tu}^{2})\sigma^{2}\\
&\ne
	(\beta_{ts}^{2}\beta_{su}^{2} +\beta_{tu}^{2}+1)\sigma^{2}
	=\text{var}(X_{t}\mid\mid C=\emptyset)
\end{align*}
We then see that the mean and variance of node $t$ are different for the observational and experimental conditions, and indicates that $p(x_{t}\mid x_{s}\mid\mid C=\emptyset)\ne p(x_{t}\mid x_{s}\mid\mid C=s)$.

\subsection{Soft interventions}\label{sec:soft-interventions}\noindent
A weaker form of intervention, discussed by \citet{Eberhardt:2007}, is one that changes the distribution of the variable on which is intervened but does not remove edges that are incident on the intervention node and leaves the causal structure intact. Such interventions are called \textit{soft} or \textit{parametric} interventions. An example is the shift intervention \citep[eg., ][]{Rothenhausler:2018} where a random variable shifts the mean and variance but does not remove connectivity with other variables. Because the structure remains the same, it is implied by a soft intervention that there is no complete control; other variables have an effect on the intervention node. The main advantage is efficiency to discover causal relations. A disadvantage is that causal effects cannot be determined as easily \citep{Eberhardt:2005}. 

Consider Figure \ref{fig:intervention-graph}(a), where we intervene on variable $X_{s}$ by $C$. Intervening does not lead to the removal of the edge $X_{u}\to X_{s}$, as was the case with a hard intervention, but the distribution of $X_{s}$ is changed with respect to no intervention. When $C=s$ in Figure \ref{fig:intervention-graph}(a) we obtain the intervention distribution
\begin{align*}
p(x\mid\mid C=s)=p(x_{t}\mid x_{u},x_{s})p(x_{s}\mid x_{u}\mid\mid C=s)p(x_{u})
\end{align*}
In a soft intervention we leave the graph structure as it is and simply identify that the observation was obtained in the context of $C=s$, and so 
\begin{align*}
p(x_{s}\mid x_{u}\mid\mid C=s)=p(x_{s}\mid x_{u},C=s)
\end{align*}
An effect can be obtained by considering whether 
\begin{align*}
p(x_{s}\mid x_{u}, C=\emptyset)=p(x_{s}\mid x_{u},C=s)\quad \text{or}\quad p(x_{s}\mid C=\emptyset)=p(x_{s}\mid C=s)
\end{align*}
It should be clear that since in Figure \ref{fig:intervention-graph}(a) the edge $X_{u}\to X_{s}$ remains, there is a confound to infer the causal effect $X_{s}\to X_{t}$. This means that soft interventions are useful to discover the structure of a graph, they are less useful in determining causal effects. 
This brings us to the the following definition. 

\begin{definition}{\em (Soft intervention)}
Let $\mathcal{G}^{a}$ be an augmented graph with system nodes $j$ in $V$ and context nodes $k$ in $K$ corresponding to context variable $C_k$, such that $k\to j$ only and not vice versa.  A {\em soft intervention} where $C_{k}\ne \emptyset$, specifies that $p(x_{j}\mid x_{\textsc{pa}(j)}\mid\mid c_{k})=p(x_{j}\mid x_{\textsc{pa}(j)},c_{k})$ such that there is some $k\in K$ with $k\to j$. All other terms in the Markov factorisation of the joint distribution remain the same as without intervention. The \emph{soft intervention distribution} is
\begin{align}\label{eq:distribution-soft-intervention}
p(x\mid\mid C=s) = \prod_{j\in V\backslash \textsc{ch}(K)}p(x_{j}\mid x_{\text{\rm pa}(j)})\prod_{i\in \textsc{ch}(K)}p(x_{i} \mid x_{\text{\rm pa}(i)},c_{k})
\end{align}
\end{definition}

Using the linear Gaussian model (\ref{eq:true-model}) in the graph of Figure \ref{fig:intervention-graph}(a), a soft intervention on node $s$ changes (e.g., shifts) $X_{s}$ by the random variable $W\in\mathbb{R}$ with non-zero expectation $\mu_{w}\ne 0$ and variance $\sigma^{2}_{w}>0$, such that  when $C=s$ we have $X_{s} =W + \beta_{su}X_{u}+\varepsilon_{s}$.
Then we find for the expectation of node $t$
\begin{align*}
\expt(X_{t}\mid\mid C=s)=\beta_{ts}\mu_{w}\ne 0=\expt(X_{t}\mid\mid C=\emptyset)
\end{align*}
and for the variance
\begin{align*}
\text{var}(X_{t}\mid\mid C=s)&=\beta_{ts}^{2}\sigma^{2}_{w}+(\beta_{ts}^{2}\beta_{su}^{2} +\beta_{tu}^{2}+1)\sigma^{2}\\
&\ne
	(\beta_{ts}^{2}\beta_{su}^{2} +\beta_{tu}^{2}+1)\sigma^{2}
	=\text{var}(X_{t}\mid\mid C=\emptyset)
\end{align*}
We then see that the structure of the graph remains unchanged and that the conditional distributions of the different environments have changed, $p(x_{t}\mid x_{s}\mid\mid C=s )\ne p(x_{t}\mid x_{s}\mid\mid C=\emptyset)$. 

Both the hard and soft intervention induce a change in the distribution, and hence result in observable differences. Of course, differences in the population need to outweigh the noise level in the measurements to be able to be detected. But in principle, both a soft and a hard intervention lead to observable changes in the distribution that may lead to a causal conclusion. 

\section{Appendix: Calculations of conditional distributions and correlations from Section \ref{sec:cond-cor}}\label{sec:appendix-calculations-correlations}\noindent
Consider the graph $s\to u\to t$, with associated variables $X_{s}$, $X_{u}$ and $X_{t}$. Using the linear model we have the equations 
\begin{align*}
X_s = e_s		\qquad X_u = \beta_{us}X_s + e_u	\qquad	X_t = \beta_{tu}X_u + e_t
\end{align*}
We are assuming here that the residuals $e_s$, $e_u$ and $e_t$ are independent and are $N(0,1)$ random variables, so that $X_s$, $X_u$ and $X_t$ are normally distributed. 

We start with rewriting the equations, which yields
\begin{align*}
X_s = e_s		\qquad X_u = \beta_{us}e_s + e_u	\qquad	X_t = \beta_{tu}\beta_{us}e_s + \beta_{tu}e_u + e_t
\end{align*}
The distribution $X_t\mid\mid C=\emptyset$ for the observational context is then \citep[see, e.g.,][Chapter 5]{BilodeauBrenner99}
\begin{align*}
X_t\mid\mid C=\emptyset \sim N(0, \beta_{tu}^2\beta_{us}^2+\beta_{tu}^2+1)
\end{align*}
In the example we have $\beta_{tu}=0.9$ and $\beta_{us}=1.8$, so that we obtain in the observational context the mean 0 and variance 4.4344. For the hard intervention, we have that $X_s=W=2+N(0,\sigma^2_w)$, and so $\mathbb{E}(X_t\mid\mid C=s)=2\beta_{tu}\beta_{us}$ (where the 2 is obtained from $W$) and variance $ \beta_{tu}^2\beta_{us}^2+\beta_{tu}^2\sigma^2_w+1$. With the data from the example, where $\sigma^2_w=1$, we the obtain the mean with the intervention on $s$, $3.24$, and the variance is 4.4344.

Next we discuss the correlation.We start with the observational context $C=\emptyset$. For the covariance between $X_s$ and $X_t$, using the independence between $e_s$, $e_u$ and $e_t$, and the variance being 1 for each, gives
\begin{align*}
\text{cov}(X_{s},X_{t}\mid\mid C=\emptyset) = \beta_{tu}\beta_{us}
\end{align*}
Noting that the variance of $X_s$ is 1 and that of $X_t$ is $ \beta_{tu}^{2}\beta_{us}^{2}+\beta_{tu}^{2}+1 $, gives the correlation
\begin{align*}
\text{cor}(X_{s},X_{t}\mid\mid C=\emptyset) = \frac{\beta_{tu}\beta_{us}} {\sqrt{ \beta_{tu}^{2}\beta_{us}^{2}+\beta_{tu}^{2}+1 }}
\end{align*}
With the data of the example, we obtain a correlation in the observational context of 0.7693. With a hard intervention on $s$ (where $X_{s}$ is replaced by $W= 2 + N(0,\sigma^{2}_{w})$) we obtain the correlation 
\begin{align*}
\text{cor}(X_{s},X_{t}\mid\mid C=s) = \frac{\beta_{tu}\beta_{us}\sigma_{w}} {\sqrt{ \beta_{tu}^{2}\beta_{us}^{2}\sigma_{w}^{2}+\beta_{tu}^{2}+1 }}
\end{align*}
We see that the correlation in the observational context equals that of the interventional context when we standardise $W$ such that $\sigma^{2}_{w}=1$. Note that a deterministic hard intervention leads to an undefined correlation. With the data from the example we obtain the same correlation in the observational as in the hard intervention context (because we chose $\sigma^2_w=1$).

Similarly, consider now the graph $s\leftarrow u \to t$. We have the equations 
\begin{align*}
X_s = \beta_{su}X_u + e_s		\qquad X_u =  e_u	\qquad	X_t = \beta_{tu}X_u + e_t
\end{align*}
or, applying substitution
\begin{align*}
X_s = \beta_{su}e_u + e_s		\qquad X_u =  e_u	\qquad	X_t = \beta_{tu}e_u + e_t
\end{align*}
The mean of $X_t$ in the observational context $C=\emptyset$ is 0 and the variance $\beta_{tu}^2$, which for the data in the example is 0.81. For the intervention $C=s$ with $X_s=2+N(0,\sigma^2_w)$ we obtain a mean of 0 for $X_t$. This is because $X_s$ does not function in $X_t$ in the graph $s\leftarrow u\to t$. 

Continuing with the correlation, the covariance when $C=\emptyset$ is 
\begin{align*}
\text{cov}(X_{s},X_{t}\mid\mid C=\emptyset) = \beta_{su}\beta_{tu}
\end{align*}
But applying a hard intervention, substituting $W=2+N(0,\sigma^2_w)$ for $X_s$ gives a covariance of 0 because $W$ is independent of $X_u=e_u$ and so also independent of $X_t = \beta_{tu}e_u + e_t$. Hence, a hard intervention on $s$ applied to the graph $s\leftarrow u\to t$, results in the correlation  $\text{cor}(X_{s},X_{t}\mid\mid C=s) =0$.

\section{Appendix: Conditional correlation}\label{sec:appendix-conditional-correlation}\noindent
The conditional correlation is the correlation obtained by considering several contexts simultaneously, like $C=\emptyset$ and $C=s$. (The case of a single context was discussed in Appendix \ref{sec:appendix-interventions}.) 

The probability distribution in the case of two or more contexts is a mixture distribution in general. For example, suppose we have the  distribution $\mathbb{P}$ and that we consider two contexts, $C=\emptyset$ and $C=s$. In this case the mixture distribution for $X_s=x_s$ and $X_t=x_t$ is 
\begin{align*}
\mathbb{P}(x_s,x_t\mid\mid \{\emptyset,s\}) = \pi_\emptyset \mathbb{P}(x_s,x_t\mid\mid C=\emptyset) + \pi_s \mathbb{P}(x_s,x_t\mid\mid C=s)
\end{align*}
wehere $\pi_\emptyset\ge 0$, $\pi_s\ge 0$ and $\pi_\emptyset+\pi_s=1$. The mixture coefficients $\pi_\emptyset$ and $\pi_s$ can be interpreted as the proportion of observations in each of the respective contexts $C=\emptyset$ and $C=s$ \citep[see also][Section 3.2]{Peters:2015}.  

With this mixture probability the mean $\mu_{s}(\{\emptyset,s\})$ across the two different contexts is
\begin{align*}
\mathbb{E}(X_s\mid\mid \{\emptyset,s\})
&=\int_{\mathbb{R}^2} x_s d\mathbb{P}(x_s\mid\mid \{\emptyset,s\})\\
&= \pi_\emptyset \int_{\mathbb{R}^2} x_s d\mathbb{P}(x_s\mid\mid C=\emptyset) + \pi_s \int_{\mathbb{R}^2} x_s d\mathbb{P}(x_s\mid\mid C=s)\\
 &=\pi_\emptyset\mathbb{E}(X_s\mid\mid C=\emptyset) +  \pi_s \mathbb{E}(X_s\mid\mid C=s)
\end{align*}
The conditional covariance for the two contexts is then
\begin{align*}
\text{cov}(X_s,X_t\mid\mid \{\emptyset,s\}) = \int_{\mathbb{R}^2} (x_s-\mu_s(\{\emptyset,s\}))(x_t-\mu_t(\{\emptyset,s\})) d\mathbb{P}(x_s,x_t\mid\mid \{\emptyset,s\})
\end{align*}
where $\mu_s(\{\emptyset,s\})$ is the mean of $X_s$ based on the contexts $C=\emptyset$ and $C=s$. Applying the mixture distribution above yields
\begin{align*}
\text{cov}(X_s,X_t\mid\mid \{\emptyset,s\})
&=\mathbb{E}(X_sX_t\mid\mid \{\emptyset,s\}) - \mathbb{E}(X_s\mid\mid \{\emptyset,s\})\mathbb{E}(X_t\mid\mid \{\emptyset,s\})\\
&=\pi_\emptyset\left(\text{cov}(X_s,X_t\mid\mid \emptyset)+\mathbb{E}(X_s\mid\mid \emptyset)\mathbb{E}(X_t\mid\mid \emptyset)\right)\\
&\quad + \pi_s\left(\text{cov}(X_s,X_t\mid\mid s)+\mathbb{E}(X_s\mid\mid s)\mathbb{E}(X_t\mid\mid s)\right)\\
&\quad -\mathbb{E}(X_s\mid\mid \{\emptyset,s\})\mathbb{E}(X_t\mid\mid \{\emptyset,s\})
\end{align*}
The variance (and standard deviation) is defined similarly with $t=s$, and then the correlation is obtained in the usual way 
\begin{align*}
\rho_{ts\mid\mid \{\emptyset,s\}}=\frac{\text{cov}(X_s,X_t\mid\mid \{\emptyset,s\})} { \text{sd}(X_s\mid\mid \{\emptyset,s\}) \text{sd}(X_t\mid\mid \{\emptyset,s\}) }
\end{align*}
%

\section{Appendix: Additional lemmas}\label{sec:appendix-lemmas}\noindent

\begin{lemma}{\rm (Reichenbach's principle)}\label{lem:direct-path}
Let $\mathcal{G}$ be a directed acyclic graph with probability distribution $\mathbb{P}$ that is Markov and faithful to $G$. Suppose that for nodes $s$ and $t$ in $\mathcal{G}$ there is a non-zero correlation $\rho_{st}\ne 0$. Then there is a directed path $s\to \cdots \to t$  from $s$ to $t$, or $s\leftarrow \cdots \leftarrow t$ from $t$ to $s$, where the intermediate node set could be empty (a direct connection), or there is a common cause such that $s\leftarrow \cdots \to t$. 
\end{lemma}
\begin{proof}
Assume we obtain a correlation $\rho_{st}\ne 0$ in the directed acyclic graph $G$. Because we assume faithfulness (Assumption \ref{ass:faith}) this implies that there must be some path between $s$ and $t$. If there is a direct connection $s\to t$ or $s\leftarrow t$, then we are done. Suppose then that there is some path between $s$ and $t$ with at least one intermediate node. Towards a contradiction, suppose that on this path between $s$ and $t$ there is a collider $s- \circ\to a \leftarrow\circ - t$. By the Markov condition (Assumption \ref{ass:markov}), this implies the independence relation $X_{s}\independent X_{t}$. But this contradicts that we already have $\rho_{st}\ne 0$. The only remaining possibilities for a DAG are 
$s\to \cdots  \to t$ or $s\leftarrow  \cdots \leftarrow t$ or $s\leftarrow\cdots\to t$.
\end{proof}
\begin{lemma}\label{lem:wiggle-s}
Let $\mathcal{G}$ be a directed acyclic graph with probability distribution $\mathbb{P}$ that is Markov and faithful. If there is a non-zero correlation $\rho_{ts}$ between $s$ and $t$, and if we intervene on node $s$ in $\mathcal{G}$, then we observe a change in the distribution of $t$ only if there is some directed path $s\to\cdots\to t$. 
\end{lemma}
\begin{proof}
By Lemma \ref{lem:direct-path}, the faithfulness and Markov conditions imply that are three options upon the finding of a non-zero correlation $\rho_{st}$. A hard (or soft) intervention will result in a change in the conditional distribution such that $p(x_{t}\mid x_{s}\mid\mid C=s)\ne p(x_{t}\mid x_{s}\mid\mid C=\emptyset)$. By the faithfulness condition, this implies that there must exist a directed path $s\to \cdots \to t$. 
\end{proof}
\begin{lemma}\label{lem:path-prod}
Let $\mathcal{G}$ be a directed acyclic graph that is Markov and faithful (see Appendix \ref{sec:appendix-gm-assumptions}) with respect to distribution $\mathbb{P}$. Assume the random variables corresponding to the nodes in the graph $\mathcal{G}$ have been standardised. Suppose for some pair of nodes $s$ and $t$ that the conditional correlation $\rho_{ts\mid\mid \{\emptyset,s\}}\ne 0$.  Then 
we have that (\ref{eq:tr-criterion}) is necessary but not sufficient to conclude that $s$ is not a parent of $t$.
\end{lemma}
\begin{proof}
We begin with sufficiency. In a directed acyclic graph $\mathcal{G}$, suppose there are two directed paths $s\to \cdots \to t$ between $s$ and $t$. We refer to these paths as $\alpha_{1}$ and $\alpha_{2}$, where $\alpha=(s,v_{2},\ldots,v_{k-1},t)$ of length $k$, for some set of $k-2$ intermediate nodes $v_{j}$. 
We use a result by \citet[][see also \cite{Wright:1934}]{Wright:1921} that characterises a correlation as the sum of all the paths between $s$ and $t$ where each directed path between $s$ and $t$ equals the product of path coefficients $\beta_{ij}$ on that path, in this case of two paths $\rho_{ts\mid\mid \{\emptyset,s\}}=d_{\alpha_{1}}+d_{\alpha_{2}}$, where $d_{\alpha}$ is defined by (\ref{eq:path-coefficient-product}). By Lemma \ref{lem:wiggle-s} only a directed path from $s$ to $t$ will contribute to the conditional correlation, and any common cause $s\leftarrow \cdots \to t$ or common effect $s\to \cdots  \leftarrow t$ will not. 

Suppose additionally that there is a direct connection $s\to t$ with coefficient $\beta_{ts}^{*}$. Then the correlation between $s$ and $t$ is $\rho_{ts\mid\mid \{\emptyset,s\}}=\beta_{ts}^{*}+d_{\alpha_{1}}+d_{\alpha_{2}}$. It is now easily seen that whenever 
\begin{align*}
|\beta_{ts}^{*}+d_{\alpha_{1}}+d_{\alpha_{2}}|< \min\{|\rho_{v_{i}v_{j}\mid\mid \{\emptyset,s\}}|: (v_{i},v_{j})\in \mathcal{P}(s\to t)\}
\end{align*}
where $\mathcal{P}(s\to t)$ is the set of edges of all directed paths from $s$ to $t$, the direct connection $s\to t$ with path coefficient $\beta^{*}_{st}<\rho_{ts\mid\mid \{\emptyset,s\}}$ remains undetected by criterion (\ref{eq:tr-criterion}). Only for values where the reverse holds will criterion (\ref{eq:tr-criterion}) work. It is obvious that this argument remains valid if the number of paths between $s$ and $t$ is increased to any finite $m$.

We continue with necessity. Suppose there is no direct connection $s\to t$ but there is a directed path from $s$ to $t$, $\alpha=(s,v_{2},\ldots,v_{k-1},t)$ of length $k$. Without loss of generality we can assume the variables to be standardised so that $\beta_{ij}=\rho_{ij}$. Then from (\ref{eq:path-coefficient-product}) it follows that $\rho_{ts\mid\mid \{\emptyset,s\}}=\prod_{j} \rho_{v_{j}v_{j-1}\mid\mid\{\emptyset,s\}}$, and so $|\rho_{v_{j-1}v_{j}\mid\mid \{\emptyset,s\}}|\ge |\rho_{ts\mid\mid \{\emptyset,s\}}|$ for all $j$ on the path, giving criterion (\ref{eq:tr-criterion}).
\end{proof}
%

\section{Appendix: Proofs of main results}\label{sec:appendix-proofs}\noindent

\vspace{1em}\noindent\textbf{Proposition \ref{prop:condcor-invpred}}\hspace{.2em}{\em
Assume the Markov assumption \ref{ass:markov} and the faithfulness assumption \ref{ass:faith} are satisfied for a DAG $G$. Suppose the variables in graph $\mathcal{G}$ are measured in contexts $C=\emptyset$ and $C=s$ of equal probability with an intervention on $s$, and the variances of the nodes $s$ and $t$ in $\mathcal{G}$ are the same across contexts. Then the following are equivalent
\begin{itemize}
\item[(i)]  $s$ is marginal invariant predictive for $t$, and
 \item[(ii)] the correlations without and with intervention are equal, i.e., $\rho_{ts\mid\mid \emptyset}=\rho_{ts\mid\mid s}$.
\end{itemize}
As a consequence, if $\rho_{ts\mid\mid \emptyset}=\rho_{ts\mid\mid s}$, then this also equals $\rho_{ts\mid\mid \{\emptyset,s\}}$
}
\begin{proof}
(i) $\to$ (ii) Suppose $s$ is marginal invariant predictive for $t$. Then for contexts $C=\emptyset$ and $C=s$ we have that 
$X_{t}=X_{s}\beta_{ts}^{*}+ \varepsilon$ for some non-zero $\beta^{*}_{ts}$. Because $\text{sd}(X_{u}\mid\mid C=\emptyset)=\text{sd}(X_{u}\mid\mid C=s)$ for nodes $s,t\in V$, we have that 
\begin{align*}\label{eq:equal-cor}
\rho_{ts\mid\mid\emptyset}=\beta^{*}_{ts}\frac{\text{sd}(X_{s}\mid\mid C=\emptyset)}{\text{sd}(X_{t}\mid\mid C=\emptyset)}=\beta^{*}_{ts}\frac{\text{sd}(X_{s}\mid\mid C=s)}{\text{sd}(X_{t}\mid\mid C=s)}=\rho_{ts\mid\mid s}
\end{align*}
This shows that if $\rho_{ts\mid\mid \emptyset}=\rho_{ts\mid\mid s}$ then $s$ is marginal invariant predictive for $t$.

(ii) $\to$ (i) Conversely, suppose that $\rho_{ts\mid\mid \emptyset}=\rho_{ts\mid\mid s}$ and that there are two different coefficients for the two contexts, $\beta_{ts}^*(\emptyset)\ne \beta_{ts}^*(s)$. Then 
\begin{align*}
\rho_{ts\mid\mid\emptyset}=\beta^{*}_{ts}(\emptyset)\frac{\text{sd}(X_{s}\mid\mid C=\emptyset)}{\text{sd}(X_{t}\mid\mid C=\emptyset)}
\end{align*}
and
\begin{align*}
\rho_{ts\mid\mid s}=\beta^{*}_{ts}(s)\frac{\text{sd}(X_{s}\mid\mid C=s)}{\text{sd}(X_{t}\mid\mid C=s)}
\end{align*}
Since we assumed that $\rho_{ts\mid\mid \emptyset}=\rho_{ts\mid\mid s}$ and that $\text{sd}(X_{u}\mid\mid C=\emptyset)=\text{sd}(X_{u}\mid\mid C=s)$ for nodes $s$ and $t$, we obtain that $\beta_{ts}^*(\emptyset)= \beta_{ts}^*(s)$.

Finally, for the consequence of this equivalence, suppose that $\rho_{ts\mid\mid \emptyset}=\rho_{ts\mid\mid s}$. Assume without loss of generality that the variables $X_s$ and $X_t$ are standardized. Then the mean of variable $X_s$ is $\mu_s(C=c)=0$, for contexts $C=\emptyset$ and $C=s$ for any nodes $s$ and $t$. Then the mixture mean $\mu_s(\{\emptyset,s\})=0$ also. And the standard deviation $\text{sd}(X_s\mid\mid C=c)=1$ for $C=\emptyset$ and $C=s$. Then from the conditional correlation in Appendix \ref{sec:appendix-conditional-correlation}, this implies that 
\begin{align*}
\text{cov}(X_s,X_t\mid\mid \{\emptyset,s\})
&=\pi_\emptyset\text{cov}(X_s,X_t\mid\mid C=\emptyset)  + \pi_s \text{cov}(X_s,X_t\mid\mid C=s)
\end{align*}
When assuming $\rho_{ts\mid\mid \emptyset}=\rho_{ts\mid\mid s}$ we immediately obtain that
\begin{align*}
\text{cov}(X_s,X_t\mid\mid \{\emptyset,s\})
&=(\pi_\emptyset + \pi_s) \text{cov}(X_s,X_t\mid\mid C=s)=\text{cov}(X_s,X_t\mid\mid C=s)
\end{align*}
which was to be shown.
\end{proof}
%

\section{Appendix: Path analysis by Wright}\label{sec:appendix-path-analysis}\noindent
We provide some details here on the calculations of the path contributions to the correlation between nodes $s$ and $t$ for the small graphs in Figure \ref{fig:wright-graph}. We have the system of equations 
\begin{align*}
X_{s}&=\varepsilon_{s}				&X_{v}&=\beta_{vs}X_{s}+\varepsilon_{v}\\
X_{w}&=\beta_{ws}X_{s}+\varepsilon_{w}	&X_{t}&=\beta_{tv}X_{v}+\beta_{tw}X_{w}+\varepsilon_{t}
\end{align*}
where the $\varepsilon_{j}$ are assumed to be independent and identical normally distributed with mean 0 and variance 1. Then, each variable $X_{j}$ has mean 0, and 
\begin{align*}
\rho_{ij}=\frac{\mathbb{E}(X_{i}-\mu_{i})(X_{j}-\mu_{j})}{\sqrt{\text{var}(X_{i})\text{var}(X_{j})}} =\mathbb{E}(X_{i}X_{j}) 
\end{align*}
as $\mu_{i}=0$ and $\text{var}(X_{i})=1$ for both $i$ and $j$. So, for the correlation between nodes $s$ and $v$ (see also Figure \ref{fig:wright-graph}(a))
\begin{align*}
\rho_{vs} = \mathbb{E}(X_{s}X_{v})=\mathbb{E}(X_{s}(\beta_{vs}X_{s}+\varepsilon_{v}))
=\beta_{vs}\mathbb{E}(X_{s}^{2})+\mathbb{E}(X_{s}\varepsilon_{v})
=\beta_{vs}
\end{align*}
by the assumptions that $\mathbb{E}(X_{s}^{2})=1$ and $\mathbb{E}(X_{s}\varepsilon_{v})=\mathbb{E}(\varepsilon_s\varepsilon_v)=0$.
Then for the correlation between $s$ and $t$ we have
\begin{align*}
\rho_{ts}=\mathbb{E}(X_{s}X_{t})=\mathbb{E}(X_{s}(\beta_{tv}X_{v}+\beta_{tw}X_{w}+\varepsilon_{t}))
\end{align*}
Because the residuals $\varepsilon_{j}$ are independent normal with mean 0 we obtain $\mathbb{E}(\varepsilon_{j}\varepsilon_{i})=0$ for any $i$ and $j$ (also $i=j$). Subsequently, we fill in for $X_{v}=\beta_{vs}X_{s}+\varepsilon_{v}$ and $X_{w}=\beta_{ws}X_{s}+\varepsilon_{w}$. Then we obtain 
\begin{align*}
\mathbb{E}(X_{s}(\beta_{tv}X_{v}+\beta_{tw}X_{w}+\varepsilon_{t}))=\beta_{tv}\beta_{vs}\mathbb{E}(X_{s}^{2})+\beta_{tw}\beta_{ws}\mathbb{E}(X_{s}^{2})
\end{align*}
Noting that $\mathbb{E}(X_{s}^{2})=1$ by assumption, we obtain the result. 

An alternative formulation for a DAG $\mathcal{G}$ is given in \citet{Shojaie:2010}. Construct an adjacency matrix $A$ such that the coefficient for the edge $s\to t$ is element $A_{ts}$ of the $p\times p$ matrix $A$, such that $A$ is lower diagonal. (This requires that the ordering of the variables be known.) Then $\Lambda=(I_p - A)^{-1}$ is the \textit{influence matrix} for the equation $X=\Lambda Z$, where $Z$ is a $p$-dimensional vector of random variables and $I_p$ is the identity matrix of dimension $p$.


\section{Appendix: Calculations of example in Section \ref{sec:cond-inv-pred}}\label{sec:appendix-calculations-example}\noindent
We consider the graph in Figure \ref{fig:perturb-graph}(a), where we have the following system of equations using the Gaussian  model (\ref{eq:true-model})
\begin{align*}
X_{v} &= \varepsilon_{v}
	&X_{u}&=\beta_{us}X_{s}+\varepsilon_{u}\\
X_{s} &= \beta_{sv}X_{v} + \varepsilon_{s}
	&X_{t}&=\beta_{tu}X_{u}+\varepsilon_{t}
\end{align*}
We assume that the variance of $\varepsilon_j$ is 1, for all $j$.  If we intervene on $s$ we obtain $X_{s}=W$ for a hard intervention and $X_{s}=W+\beta_{sv}X_{v} + \varepsilon_{s}$ for a soft intervention. 
Then we find that if we determine the covariance between $X_{s}$ and $X_{t}$ in the observation context $C=\emptyset$ we obtain (see below for the details of the calculations) 
\begin{align*}
\text{cov}(X_{s},X_{t}\mid\mid C=\emptyset) = \beta_{sv}^{2}\beta_{tu}\beta_{us}+\beta_{tu}\beta_{us}
\end{align*}
When we intervene on node $s$ we obtain using a hard intervention
\begin{align*}
\text{cov}(X_{s},X_{t}\mid\mid C=s) = \beta_{tu}\beta_{us}\sigma^{2}_{w}
\end{align*}
Similarly, for a soft intervention we find
\begin{align*}
\text{cov}(X_{s},X_{t}\mid\mid C=s) = \beta_{tu}\beta_{us}\sigma^{2}_{w} + \beta_{sv}^{2}\beta_{tu}\beta_{us}+\beta_{tu}\beta_{us}
\end{align*}
Hence, the conclusion using the conditional correlation in (\ref{eq:cond-cor-sample}) with either a hard or soft intervention will lead to a non-zero correlation. However, if we condition on $X_{u}=x_{u}$ then, with or without intervention, we obtain a covariance between nodes $s$ and $t$ of 0, implying a correlation of 0. In this example, when conditioning on $X_{u}$ we obtain the conditional distribution $X_{t}\mid x_{u}$, which is $X_{t}-\beta_{tu}x_{u}=\varepsilon_{t}$. And since we assumed that $\varepsilon_{t}$ is uncorrelated with $\varepsilon_{s}$ (in the Gaussian model (\ref{eq:true-model})) we obtain a correlation of $0$. This is shown in Figure \ref{fig:scatterplots-invariant-prediction}(b), where the observational and the interventional scatterplots show a zero correlation. So, conditioning is required for any context $C$ to obtain the parents (direct edges). In contrast, in Figure \ref{fig:perturb-graph}(b) we see that conditioning on $u$ will not affect the influence of $s$ on $t$, and so there we will find that in both contexts $C=\emptyset$ and $C=s$ the conditional distributions $X_{t}\mid x_{u}\mid\mid C=\emptyset$ and $X_{t}\mid x_{u}\mid\mid C=s$ will be different. Hence, $u$ is not sufficient to block the effect of $s$ on $t$ (there could, for instance, be at least two directed paths from $s$ to $t$). In Figure \ref{fig:perturb-graph}(b) we see that $u$ is not between $s$ and $t$, so that $t$ will be affected by $s$ even when conditioning on $u$.  

Here are several details of the analysis for the graphs in Figure \ref{fig:perturb-graph}(a). We use the system of equations corresponding to Figure \ref{fig:perturb-graph}(a), which for the observational context $C=\emptyset$ is
\begin{align*}
X_{v} &= \varepsilon_{v}
	&X_{u}&=\beta_{us}X_{s}+\varepsilon_{u}\\
X_{s} &= \beta_{sv}X_{v} + \varepsilon_{s}
	&X_{t}&=\beta_{tu}X_{u}+\varepsilon_{t}
\end{align*}
And $X_{s} =W$ for a hard intervention and $X_{s}=W+\beta_{sv}X_{v} + \varepsilon_{s}$ for a soft intervention, where $W$ is a random variable with mean $\mu_{w}$ and variance $\sigma^{2}_{w}$. We limit ourselves here to covariance analysis since we need to show that in certain cases the covariance (and hence the correlation) is 0. In the observational context $C=\emptyset$ the covariance between $X_{s}$ and $X_{t}$ is 
\begin{align*}
\text{cov}(X_{s},X_{t}\mid\mid C=\emptyset)&=\mathbb{E}(X_{s}X_{t}\mid\mid C=\emptyset)\\
	&=\mathbb{E}\{(\beta_{sv}\varepsilon_{v}+\varepsilon_{s})(\beta_{tu}\beta_{us}\beta_{sv}\varepsilon_{v}+\beta_{tu}\varepsilon_{u}+\beta_{tu}\beta_{us}\varepsilon_{s})\mid\mid C=\emptyset\}\\
	&=\beta_{tu}\beta_{us}\beta_{sv}^{2} +\beta_{tu}\beta_{us}
\end{align*}
where we assumed that the variance of all $\varepsilon_{j}$ is 1. When conditioning on $X_{u}=x_{u}$ we have that the covariance between $X_{s}$ and $X_{t}$ in context $C=\emptyset$ is 
\begin{align*}
\mathbb{E}(X_{s}(X_{t}-\beta_{tu}x_{u})\mid x_{u}\mid\mid C=\emptyset)=\mathbb{E}((\beta_{sv}\varepsilon_{v}+\varepsilon_{s})\varepsilon_{t}\mid x_{u}\mid\mid C=\emptyset)=0
\end{align*}
which corresponds to the graph in Figure \ref{fig:perturb-graph}(a) where there is no edge $s\to t$. In the hard intervention, including conditioning on $X_{u}=x_{u}$ we find
\begin{align*}
\mathbb{E}(X_{s}(X_{t}-\beta_{tu}x_{u})\mid x_{u}\mid\mid C=s)=\mathbb{E}((W-\mu_{w})\varepsilon_{t}\mid x_{u}\mid\mid C=s)=0
\end{align*}
since $W$ is independent of $\varepsilon_{t}$. Similarly, we find for the soft intervention when conditioning on $X_{u}=x_{u}$
\begin{align*}
\mathbb{E}\{(W-\mu_{w}+\beta_{sv}\varepsilon_{v}+\varepsilon_{s})(\beta_{tu}x_{u}+\varepsilon_{t}-\beta_{tu}x_{u})\mid x_{u}\mid\mid C=s\}=0
\end{align*}
Without conditioning on $X_{u}$ we find for the hard intervention
\begin{align*}
\mathbb{E}\{(W-\mu_{w})(\beta_{tu}\beta_{us}W-\beta_{tu}\beta_{us}\mu_{w}+\beta_{tu}\varepsilon_{u}+\varepsilon_{t})\mid\mid C=s\}=\beta_{tu}\beta_{us}\sigma^{2}_{w}
\end{align*}
And for the soft intervention
\begin{align*}
&\mathbb{E}\{(W-\mu_{w}+\varepsilon_{s})(\beta_{tu}\beta_{us}W-\beta_{tu}\beta_{us}\mu_{w}+\beta_{tu}\beta_{us}\beta_{sv}X_{v}\\
	&\qquad +\beta_{tu}\beta_{us}\varepsilon_{s}+\beta_{tu}\varepsilon_{u}+\varepsilon_{t})\mid\mid C=s\}
	=\beta_{tu}\beta_{us}\sigma^{2}_{w}+\beta_{tu}\beta_{us}\beta_{sv}^{2} + \beta_{tu}\beta_{us}
\end{align*}
This shows that that the correct inference is obtained only when conditioning on $X_{u}=x_{u}$.

\section{Appendix: Unobserved confounds from Section \ref{sec:unobserved}}\label{sec:appendix-unobserved}\noindent
We assume the graph of Figure \ref{fig:hidden-graph}(a). Here node $u$ is unobserved (hidden) and so we work with the marginal distribution over the nodes $s,t$ and a binary variable $C$ that induces a soft intervention such that $\expt(X_{s}\mid\mid C=0)\ne \expt(X_{s}\mid\mid C=1)$. 
The system of linear equations corresponding to the graph in Figure \ref{fig:hidden-graph}(a) is as follows.
\begin{align*}
X_{t} &= \beta_{ts}X_{s}+\beta_{tu}X_{u}+\varepsilon_{t}	&X_{s} &= \beta_{sC} X_{I}+\beta_{su}X_{u}+\varepsilon_{s}\\
X_{u} &=  \varepsilon_{u}				&X_{C} &= 0/1
\end{align*}
If we know node $u$ then we have that the residual of $\varepsilon_{t}=X_{t}-(X_{s},X_{u})\gamma^{*}$ and its support $(X_{s},X_{u})$ are independent, because with the correct coefficient $\gamma^{*}$
\begin{align*}
\text{cov}(X_{t}-(X_{s},X_{u})\gamma^{*},(X_{s},X_{u}))=(I-P_{su})\text{cov}((X_{s},X_{u})\gamma^{*},(X_{s},X_{u}))=0
\end{align*}
where $P_{su}=X(X^{\sf T}X)^{-1}X^{\sf T}$ is the projection matrix on the space of nodes $s$ and $u$, and $X=(X_{s},X_{u})$.
However, when $u$ is unobserved we make the residual orthogonal to $X_{s}$ only, whence $P_{s}=X_{s}(X_{s}^{\sf T}X_{s})^{-1}X_{s}^{\sf T}$ and we obtain 
\begin{align*}
\text{cov}(X_{t}-X_{s}\gamma_{s},(X_{s},X_{u}))=(I-P_{s})\gamma_{s}\sigma^{2}_{u}
\end{align*}
wher $\gamma_{s}$ is the regression coefficient with only $X_{s}$ as predictor.  This is 0 only if $\gamma_{s}=0$ or if $\sigma_{u}^{2}=0$. It follows that when there are unobserved confounders, it is not possible to demand that the residual $X_{t}-X_{s}\gamma^{*}$ and its support $(X_{s},X_{u})$ are independent.

\section{Appendix: Simulation examples \textsf{R} code}\label{sec:appendix-simulations}\noindent

{\small
\begin{verbatim}
# libraries
library(Rgraphviz)
library(pcalg)
library(qgraph)
library(ggplot2)
library(ggExtra)

# general settings
n <- 100 	# nr of samples
p <- 3 		# nr of variables 
set.seed(123)
cols <- c('#0000FF50','#FF000050')

############# model s -> u -> t
graph <- matrix(c(
			0, 1, 0,
			0, 0, 1,
			0, 0, 0
			),ncol=p,byrow=TRUE)
labels <- c("s","u","t")			
rownames(graph) <- colnames(graph) <- labels
qgraph(graph,labels=labels)

# model parameters, beta != 1
betaUS <- 1.8
betaTU <- 0.9

# graph s -> u -> t
B <- matrix(c(
			0, betaUS, 0,
			0, 0,	   betaTU,
			0, 0,	   0
			),ncol=p,byrow=TRUE)

IminB <- diag(1,p) - B

# data generation according to graph
Z <- matrix(rnorm(n*p),ncol=p,nrow=n) # n x p matrix of standard normal variables
X <- Z%*%solve(IminB) # data Markov to graph
Xi <- X
Zi <- Z
Zi[,1] <- rnorm(n) + 2 # intervention on node s same variance 
Xi <- Zi%*%solve(IminB)

# chain graph
# make into dataframe
dX <- rbind(X,Xi)
context <- c(rep('observational',n),rep('interventional',n))
dx <- data.frame(s=dX[,1],u=dX[,2],t=dX[,3],context=context)
# plot
psut <- ggplot(dx, aes(s, t, colour = context)) 
	+ geom_point() 
	+ theme(legend.position="none") 
	+ scale_color_manual(values=cols) 
	+ xlab("node s") 
	+ ylab("node t") 
	+ geom_point(aes(x=mean(X[,1]),y=mean(X[,3])),colour="red",size=3) 
	+ geom_point(aes(x=mean(Xi[,1]),y=mean(Xi[,3])),colour="blue",size=3)
ggMarginal(psut, groupColour = TRUE, groupFill = TRUE) 



# plot conditional (on u) 
Xs.ip <- residuals(lm(X[,1]~X[,2]))
Xis.ip <- residuals(lm(Xi[,1]~Xi[,2]))
Xt.ip <- residuals(lm(X[,3]~X[,2]))
Xit.ip <- residuals(lm(Xi[,3]~Xi[,2]))

# make into dataframe
dXr <- rbind(cbind(Xs.ip,Xt.ip),cbind(Xis.ip,Xit.ip))
context <- c(rep('observational',n),rep('interventional',n))
dxs <- data.frame(s=dXr[,1],t=dXr[,2],context=context)
# plot
psut <- ggplot(dxs, aes(s, t, colour = context)) 
	+ geom_point() 
	+ theme(legend.position="none") 
	+ scale_color_manual(values=cols) 
	+ xlab("node s") 
	+ ylab("node t") 
	+ geom_point(aes(x=mean(Xs.ip),y=mean(Xt.ip)),colour="red",size=3) 
	+ geom_point(aes(x=mean(Xis.ip),y=mean(Xit.ip)),colour="blue",size=3)
ggMarginal(psut, groupColour = TRUE, groupFill = TRUE)


# make into dataframe
dXit <- rbind(Xit.ip,Xt.ip)
context <- c(rep('observational',n),rep('interventional',n))
dx <- data.frame(s=dX[,1],u=dX[,2],t=dX[,3],context=context)
# plot
psut <- ggplot(dx, aes(s, t, colour = context)) 
	+ geom_point() 
	+ theme(legend.position="none") 
	+ scale_color_manual(values=cols) 
	+ xlab("node s") 
	+ ylab("node t") 
	+ geom_point(aes(x=mean(s),y=mean(t)),colour="blue",size=3) 
	+ geom_point(aes(x=0,y=0),colour="red",size=3)
ggMarginal(psut, groupColour = TRUE, groupFill = TRUE)


############# common cause model #############
set.seed(123)
# graph s <- u -> t
graph2 <- matrix(c(
			0, 0, 0,
			1, 0, 1,
			0, 0, 0
			),ncol=p,byrow=TRUE)
labels <- c("s","u","t")			
rownames(graph2) <- colnames(graph) <- labels
qgraph(graph2,labels=labels)

# model parameters, beta != 1
betaSU <- 1.8
betaTU <- 0.9

# model s <- u -> t
B2 <- matrix(c(
			0, 0, 0,
			betaSU, 0, betaTU,
			0, 		0, 0
			),ncol=p,byrow=TRUE)

IminB2 <- diag(1,p) - B2

# data generation according to graph s <- u -> t
Z2 <- matrix(rnorm(n*p),ncol=p,nrow=n) # n x p matrix of standard normal variables
X2 <- Z2%*%solve(IminB2) # data Markov to graph
X2i <- X2
X2i[,1] <- rnorm(n) + 2 # intervention  on s same variance

# make into dataframe
dX2 <- rbind(X2,X2i)
context <- c(rep('observational',n),rep('interventional',n))
dx2 <- data.frame(s=dX2[,1],u=dX2[,2],t=dX2[,3],context=context)
# plot
psut2 <- ggplot(dx2, aes(s, t, colour = context)) 
	+ geom_point() + theme(legend.position="none") 
	+ xlab("node s") 
	+ ylab("node t") 
	+ scale_color_manual(values=cols) 
	+ geom_point(aes(x=mean(X2[,1]),y=mean(X2[,3])),colour="red",size=3) 
	+ geom_point(aes(x=mean(X2i[,1]),y=mean(X2i[,3])),colour="blue",size=3)
ggMarginal(psut2, groupColour = TRUE, groupFill = TRUE)
\end{verbatim}
}

\end{document}